\documentclass[fontsize=11pt, headings=normal, headinclude=true, paper=a4, pagesize, cleardoublepage=current, BCOR=10mm, fleqn, numbers=noenddot]{scrartcl}
\pdfoutput=1
\usepackage{lmodern}

\setkomafont{sectioning}{\normalcolor\bfseries}
\recalctypearea 
\setcapindent{0pt}

  \usepackage[english]{babel}  
  \usepackage[numbers,square]{natbib}
  \usepackage{mathrsfs}
  \usepackage{amsmath}
  \usepackage{amssymb}
  \usepackage{amsthm}
  \usepackage{enumitem}
  \usepackage[utf8]{inputenc}    
  \usepackage[T1]{fontenc}         
  \usepackage{graphicx}            
  \usepackage{color}		
  \usepackage[english=british]{csquotes} 

\newtheorem{thm}{Theorem}[]
\newtheorem{prop}[thm]{Proposition}
\newtheorem{lem}[thm]{Lemma}
\newtheorem{cor}[thm]{Corollary}

\theoremstyle{definition}

\newtheorem{rem}[thm]{Remark}

\renewcommand{\phi}{\varphi}
\newcommand{\eps}{\varepsilon}

\newcommand{\ud}{\mathrm{d}}
\newcommand{\ue}{\mathrm{e}}
\newcommand{\ui}{\mathrm{i}}
\newcommand{\R}{\mathbb{R}}

\newcommand{\N}{\mathbb{N}}

\newcommand{\norm}[1]{\ensuremath{\left\lVert #1 \right\rVert}}
\newcommand{\abs}[1]{\ensuremath{\left\lvert #1 \right\rvert}}

\newcommand{\hilb}{\mathscr{H}}
\newcommand{\uppar}[1]{\ensuremath{^{(#1)}}}

\title{A nonrelativistic quantum field theory with point interactions in three dimensions}

\author{Jonas Lampart 
\thanks{CNRS \& Laboratoire interdisciplinaire Carnot de Bourgogne (UMR 6303), Université de Bourgogne Franche-Comté, 9 Av. A. Savary, 21078 Dijon Cedex, France.
\texttt{jonas.lampart@u-bourgogne.fr}}
}

\begin{document}

\maketitle  

\begin{abstract}

We construct a Hamiltonian for a quantum-mechanical model 
of nonrelativistic particles in three dimensions interacting  via the creation and annihilation of a second type of nonrelativistic particles, which are bosons. The interaction between the two types of particles is a point interaction concentrated on the points in configuration space where the positions of two different particles coincide. 
We define the operator, and its domain of self-adjointness, in terms of co-dimension-three boundary conditions on the set of collision configurations relating sectors with different numbers of particles.
\end{abstract}

\section{Introduction}

In this article we introduce a mathematical model for an interacting system composed of two kinds of nonrelativistic quantum particles.
The number of the first kind of particles, which we call $x$-particles, is conserved. These particles interact by creating and annihilating bosons of a second kind, called $y$-particles. 
The interaction is supported by the set of collision configurations between at least one $x$-particle and one $y$-particle.
Formally, it is given by the linear coupling $a(\delta_x)+a^*(\delta_x)$, where $\delta_x$ denotes the delta-distribution at the position $x$ of an $x$-particle and $a^*,a$ are the creation and annihilation operators of the $y$-particles. 

A very similar model is the BEC polaron, used in physics to describe the interaction of impurities, the $x$-particles, with a dilute Bose-Einsetin condensate~\cite{grusdt2015, grusdt2016, vlietinck2015}. In this context, it is natural to assume point interactions due to the dilute nature of the system. Furthermore, in the Bogoliubuv approximation, the quasi-particles describing excitations of the condensate have a dispersion relation $\omega(k)=\sqrt{a k^2 + bk^4}$ that grows quadratically for large momenta, so they behave simlarly to nonrelativistic particles.

Our model is also closely related to the local, Galilean invariant, Lee model~\cite{levy1967, schrader1968}. In that model, there are three types of nonrelativistic particles, usually called $V$, $\Theta$ and $N$. The $V$-particles can create $\Theta$-particles, whereupon they are transformed into $N$-particles. Conversely, a $\Theta$ and an $N$-particle can combine to form a $V$-particle. These processes are such that the total mass is conserved, $m_V=m_\Theta+m_N$. Since the $N$ and $\Theta$-particles cannot create any further ones, this model is composed of invariant sectors with a bounded total number of particles.
This, and the constraint on the masses, are the essential differences to the model we will consider.

Because of the singular interaction, the formal Hamiltonian of such a model is ultraviolet divergent and ill defined. 
We will define a self-adjoint and bounded-from-below Hamiltonian for our model and describe its domain using (generalised) boundary conditions that relate the (singular) behaviour of the wavefunction near the collision configurations and the wavefunction with one $y$-particle less. 
As a general approach to the ultraviolet problem, such boundary conditions were proposed by Teufel and Tumulka~\cite{TeTu15, TeTu16}, who called them interior-boundary conditions. 
A similar method was applied by Thomas~\cite{thomas1984} to a model with one or two $x$-particles and at most one $y$-particle, which closely resembles one of the sectors of the Galilean Lee model.
A variant of our model, where the $x$-particles are not dynamical but fixed at certain locations, was treated by Schmidt, Teufel, Tumulka and the author~\cite{IBCpaper}. Schmidt and the author~\cite{LaSch18} showed that one can define the Hamiltonian for less singular models with dynamical $x$-particles using this type of boundary conditions. Examples include the Nelson model and the two-dimensional variant of our problem. These results were generalised by Schmidt to general dispersion relations for the $x$-particles~\cite{schmidt2018}, and massless models~\cite{schmidt2019}. We expect that our analysis can be similarly generalised.

For many models in nonrelativistic quantum field theory the Hamiltonian may be defined by a renormalisation procedure. 
To our knowledge, no rigorous method was known to work for our problem, so far.
It was was observed numerically in~\cite{grusdt2015} (for a model with the same ultraviolet behaviour as ours) that, after subtracting the expected (linearly) divergent quantity, there is still a (logarithmic) divergence in the ultraviolet cutoff. 
Concerning rigorous results, the less singular cases treated in~\cite{LaSch18}
can be renormalised using a technique due to Nelson~\cite{nelson1964, GrWu18} -- but this does not seem to work for the three-dimensional model (the conditions of e.g.~\cite[Thm 3.3]{GrWu18} are not satisfied, and the predicted divergence would be in conflict with the numerical results of~\cite{grusdt2015}). 
Schrader~\cite{schrader1968} used a reordered resolvent expansion to renormalise the Hamiltonian for the Galilean Lee model. However, later generalisations of this method~\cite{eckmann1970, GMW18, wuensch17} do not cover our specific model either. This seems to be related to the fact that the  constraint on the masses in the Galilean Lee model alters the structure of the singularities in a specific way (see Remark~\ref{rem:cancellations}).

In this article, we will adapt the techniques of~\cite{LaSch18} to define the Hamiltonian for our model.
We explain how this result can be understood in the language of renormalisation in Remark~\ref{rem:Dressing}.
We also give an explicit characterisation of the domain, which is generally not easy to obtain by renormalisation. 

\section{Overview and Results}

In this section we present our main results and the reasoning behind them, while leaving the technical details for the later sections.
Let us first introduce the necessary notation. We consider a fixed number $M$ of $x$-particles, so the Hilbert space of our problem is
\begin{equation*}
 \hilb:=L^2(\R^{3M})\otimes \Gamma(L^2(\R^3))\,,
\end{equation*}
where $\Gamma(L^2(\R^3))$ is the bosonic Fock space over $L^2(\R^3)$. We will denote the sector of $\hilb$ with $n$ $y$-particles by $\hilb\uppar{n}$, by $\psi\uppar{n}$ the component of $\psi\in \hilb$ in this sector and by $N=\ud \Gamma(1)$ the number operator.
The Hamiltonian for the non-interacting model is
\begin{equation}\label{eq:Hfree}
 L=-\frac{1}{2m}\sum_{\mu=1}^M \Delta_{x_\mu} + \ud \Gamma(-\Delta_y+1),
\end{equation}
where $m$ is the mass of the $x$-particles, and we have set the mass of the $y$-particles to $\frac12$ and their rest-energy to one. Its domain $D(L)$ is given by
\begin{equation*}
D(L)=\left\{ \psi \in D(N) \Big\vert \forall n\in \N: \psi\uppar{n}\in H^2(\R^{3(M+n)}) \right \}.  
\end{equation*}
The interaction operator is formally given by
\begin{equation*}
 \sum_{\mu=1}^M \left( a^*(\delta_{x_\mu})+a(\delta_{x_\mu}) \right).
\end{equation*}
The obvious problem with this operator is that the creation operator is not a densely defined operator on $\mathscr{H}$, as
\begin{equation*}
 a^*(\delta_{x_\mu})\psi\uppar{n}= \frac{1}{\sqrt {n+1}} \sum_{i=1}^n \delta_{x_\mu}(y_i)\psi\uppar{n}(X, \hat Y_i)
\end{equation*}
(where $X=(x_1, \dots, x_M)$, $ \hat Y_i=(y_1,\dots ,y_{i-1}, y_{i+1},\dots, y_{n+1})\in \R^{3n}$), is not an element of $\hilb\uppar{n+1}$ for any nonzero $\psi\uppar{n}\in \hilb\uppar{n}$. Note, however, that the creation operator is a well-defined operator from $\hilb$ to a space of distributions. The annihilation operator is less problematic, since
\begin{equation}\label{eq:a}
 a(\delta_{x_\mu})\psi\uppar{n}= \frac{1}{\sqrt {n}} \sum_{i=1}^n \psi\uppar{n}(X, Y)\vert_{y_i=x_\mu}= \sqrt{n} \psi\uppar{n}(X, Y)\vert_{y_n=x_\mu},
\end{equation}
 is well defined for $\psi \in D(L)$.

 Our approach is not to give a meaning to the interaction operator directly, but rather to implement an interaction using boundary conditions on the sets
 \begin{equation}\label{eq:Cdef}
\mathscr{C}^n=\left\{(X,Y)\in \R^{3M}\times \R^{3n} \Bigg\vert \prod_{\mu=1}^M\prod_{i=1}^n |x_\mu-y_i|=0 \right\}  
 \end{equation}
of collision configurations between the $x$-particles and the $y$-particles. 
We will then see that the resulting operator includes an interaction of the desired form when interpreted in the sense of distributions.
This is similar to the representation of point interactions, with fixed particle number, by boundary conditions, see e.g.~\cite{DeFiTe1994, Co_etal15, MoSe17}.

 In order to impose boundary conditions for $L$, we first restrict $L$ to the domain
 \begin{equation}\label{eq:DL_0}
 D(L_0)=\left\{ \psi \in D(N) \Big\vert \forall n\in \N: \psi\uppar{n}\in H^2_0(\R^{3(M+n)}\setminus \mathscr{C}^n) \right \}
 \end{equation}
of functions that vanish on $\mathscr{C}^n$. The adjoint $L_0^*$ of $L_0=L\vert_{D(L_0)}$ is then an extension of $L$. Its domain contains functions of the form 
\begin{equation}\label{eq:Gdef}
G\phi\uppar{n}:=  - \sum_{\mu=1}^M L^{-1} a^*(\delta_{x_\mu}) \phi\uppar{n}\in \hilb\uppar{n+1},
\end{equation}
with $\phi \in \hilb$.
These functions diverge like $|x_\mu-y_i|^{-1}$ near the set where $|x_\mu-y_i|=0$, see Eq.~\eqref{eq:G exp}. 
We interpret
\begin{align*}
 \big(B_\mu\psi\uppar{n+1}\big)(X,Y)=
  \beta\lim_{ y_{n+1}\to x_\mu} \abs{x_\mu-y_{n+1}} \psi\uppar{n+1}(X,Y,y_{n+1}),
\end{align*}
with some constant $\beta$ (depending on $m$ and $n$), as a singular boundary value of such functions (the limit exists almost everywhere, see Lemma~\ref{lem:B}). We set $B=\sum_{\mu=1}^M B_\mu$ and choose $\beta$ so that $B G\phi=M \phi$. The free operator $L$ then corresponds to the operator with the Dirichlet-type boundary condition $ B\psi=0$.

There is a second relevant boundary value operator, given by the finite part of $\psi\uppar{n+1}$ at $y_{n+1}=x_\mu$. More precisely
\begin{align*}
&\big(A_\mu \psi\uppar{n+1}\big)(X,Y)\\
&= \lim_{r \to 0}\frac{\sqrt{n+1}}{4\pi} \int\limits_{S^2} \left(\psi\uppar{n+1}(X,Y,x_\mu +r \omega)- \frac{\alpha}{r}(B\psi\uppar{n+1})(X,Y)\right)\ud \omega,
\end{align*}
where $\alpha=(M\beta)^{-1}$ and the limit is taken in the distributional sense, see Lemma~\ref{lem:A}.
Note that this  is a local boundary value operator extending the evaluation at $y_{n+1}=x_\mu$. We set $A=\sum_{\mu=1}^M A_\mu$, which then extends the annihilation operator~\eqref{eq:a}.

It has been shown for variants of our model (with fixed $x$-particles in~\cite{IBCpaper}, and in two dimensions in~\cite{LaSch18}) that the operator $L_0^*+A$ is self-adjoint on a domain characterised by the boundary conditions $B \psi\uppar{n+1}=M \psi\uppar{n}$, for all $n\in \N$. These conditions should be viewed as relations on the space $D(L)\oplus G \hilb \subset D(L_0^*)$.  Since 
any vector in that space can be written as $\psi = \xi + G \phi$ with unique $\xi\in D(L)$, $\phi \in \hilb$, and because $BG \phi = M \phi$, the condition $B\psi = M\psi$ is equivalent to $\psi=\phi$, or $\xi=\psi-G\psi \in D(L)$.
To see how this boundary condition is related to the creation and annihilation operators, first observe that $G\psi$ is in the kernel of $L_0^*$, for every $\psi\in \hilb$. That is, for any $\phi \in D(L_0)$, which vanishes on the set $\mathscr{C}^n$ of collision configurations by definition, we have
\begin{align}
 \langle L_0^* G\psi, \phi \rangle_\hilb &=  \langle  \psi,G^* L_0 \phi \rangle_\hilb\notag\\
 &=-\sum_{n=0}^\infty\sum_{\mu=1}^M \langle \psi\uppar{n},a(\delta_{x_\mu})\phi\uppar{n+1} \rangle_{\hilb\uppar{n}}=0.
 \label{eq:G ker}
\end{align} 
Then, taking any $\psi\in \hilb$ such that $\psi-G\psi\in D(L)$, gives 
\begin{equation}\label{eq:H_0 create}
L_0^*\psi + A\psi = L(\psi - G\psi) + A \psi = L\psi + \sum_{\mu=1}^M a^*(\delta_{x_\mu})\psi+A\psi,
\end{equation}
in the sense of distributions. Thus $L_0^*+A$ represents an operator with creation and annihilation operators up to a choice of domain, including the boundary condition, and the choice of the extension $A$. The latter choice is made in such a way that the operator is local.

In the proof of self-adjointess for the two-dimensional model~\cite{LaSch18}, a key step is then to use again that $G^*L=-\sum_{\mu=1}^M a(\delta_{x_\mu})$, and that $A$ extends the annihilation operator, to rewrite the Hamiltonian as
\begin{align}
 L_0^*\psi + A\psi&=L(1 - G)\psi + A \psi \notag\\
 &= (1-G)^*L(1-G)\psi + A \psi - \sum_{\mu=1}^M a(\delta_{x_\mu})(1-G)\psi\notag\\
 &=(1-G)^*L(1-G)\psi + AG\psi.\label{eq:H_0 trafo}
\end{align}
The operator $T=AG$ is a symmetric operator on a domain $D(T)\supset D(L)$, and the proof is then essentially reduced to proving appropriate  bounds for this operator on the domain of $L_0^*$ with the boundary condition $\psi-G\psi \in D(L)$. This operator also appears in the theory of point interactions, where it is known as the Skornyakov--Ter-Martyrosyan operator.  

In the three-dimensional case, the range of $G$ is not contained in the domain $D(T)=D(L^{1/2})$ of $T$. Since a vector $\psi$ satisfying $\psi-G\psi \in D(L)$ is in $D(L^{1/2})$ if and only if $G\psi$ is, $T$, and thus also $A$, do not map all vectors satisfying $(1-G)\psi \in D(L)$, to $\hilb$  (see also Lemma~\ref{lem:T} and Remark~\ref{rem:N=1}). Consequently, this is not a good domain for $L_0^* +A$. 
This forces us to to modify our approach, by treating $T$ as a perturbation to the free part $L$.
We let 
\begin{equation}
 K=L+T
\end{equation}
with $D(K)=D(L)$. This operator is self-adjoint and bounded from below (see Section~\ref{sect:K}). We may restrict it to the kernel of $a(\delta_x)$, as above, and define
\begin{equation*}
 K_0= K\vert_{D(L_0)}.
\end{equation*}
The domain of the adjoint $K_0^*$ contains functions of the form
\begin{equation*}
G_T\phi\uppar{n}:=  - \sum_{\mu=1}^M (K+c_0)^{-1} a^*(\delta_{x_\mu}) \phi\uppar{n} \in \hilb\uppar{n+1},
\end{equation*}
where $c_0>0$ is chosen appropriately. These functions are in the kernel of $K_0^*+c_0$, by the reasoning of Eq.~\eqref{eq:G ker}.
Their main divergence is still proportional to $|x_\mu-y_i|^{-1}$, but their asymptotic expansion also contains a logarithmically divergent term, as we will show.
Thus, the boundary operator $B$ will still be defined on such functions and we can extend $a(\delta_{x_\mu})$, by taking the finite part in this expansion, which yields an operator $A_T$. 

\begin{align*}
&\big(A_T\psi\uppar{n+1}\big)(X,Y)\\
&=
\lim_{r\to 0}\sum_{\mu=1}^M
\frac{1}{4\pi}\int\limits_{S^2}
\left(\sqrt{n+1} \psi\uppar{n+1}(X,Y,x_\mu+r\omega)+ f(r) (B\psi\uppar{n+1})(X,Y)\right) \ud \omega,\notag
\end{align*}
where $f(r)=\alpha r^{-1}+\gamma \log r$ is a fixed function, with explicit constants $\alpha$, as for $A$, and $\gamma$ (depending on $m$, see Eq.~\eqref{eq:alpha_m}). This local boundary value operator then extends the sum of the $a(\delta_{x_\mu})$ to elements of the range of $G_T$, see Proposition~\ref{prop:A_T}.
Its action there is given by
\begin{equation*}
 A_T G_T =T +S,
\end{equation*}
with an operator $S$ that is symmetric on $D(S)=D(L^\eps)$ for any $\eps>0$.

For any $\psi \in D(K_0^*)$ such that $\psi-G_T\psi\in D(L)$, we then have, as in~\eqref{eq:H_0 create},
\begin{align}
 L\psi + \sum_{\mu=1}^M a^*(\delta_{x_\mu})+A_T\psi
 &= (K+c_0)(1-G_T)\psi -c_0\psi - T\psi+ A_T\psi\notag\\
 &= K_0^*\psi + (A_T-T)\psi.\label{eq:H create}
\end{align}
The left hand side is similar to~\eqref{eq:H_0 create}, but we have chosen a different domain $D(K_0^*)$, and thus also a different extension $A_T$.
This is the expression that will replace $L_0^*+A$ in our case of the three-dimensional model. Note that $A_T-T=A_T(1-G_T) + S$, so this operator is better behaved than $A_T$ alone and it will map vectors satisfying $(1-G_T)\psi\in D(L)=D(K)$ to $\hilb$.

Using $G_T^*(K+c_0)=-\sum_{\mu=1}^M a(\delta_{x_\mu})$, we can further rewrite the Hamiltonian, as in~\eqref{eq:H_0 trafo}, and obtain
\begin{align}\label{eq:H trafo}
K_0^* + (A_T-T)&=(K+c_0)(1-G_T)+\sum_{\mu=1}^M a(\delta_{x_\mu}) (1-G_T) +S- c_0\notag\\
&=(1-G_T)^*(K+c_0)(1-G_T)+S -c_0. 
\end{align}

Our main result is:

\begin{thm}\label{thm:sa}
Let $A$ and $A_T$, be the extensions of $\sum_{\mu=1}^M a(\delta_{x_\mu})$ defined above, $AG=T$ and $A_TG_T-T=S$.
Let
\begin{align*}
 D(H)=\left\{ \psi \in D(N) \Big\vert \forall n\in \N: \psi\uppar{n+1} - G_T\psi\uppar{n}\in H^2\left(\R^{3(M+n+1)}\right) \right\}, 
\end{align*}
then the operator
 \begin{align*}
  H&=K_0^* +(A_T-T)\\
  &=(1-G_T)^*(K+c_0)(1-G_T)+S -c_0
 \end{align*}
is self-adjoint on $D(H)$ and bounded from below. Furthermore, for any $n\in \N$  and $\psi\in D(H)$ we have (setting $\psi\uppar{-1}=0$)
\begin{equation*}
 H\psi\uppar{n} = L\psi\uppar{n} +  \sum_{\mu=1}^M a^*(\delta_{x_\mu})\psi\uppar{n-1}  +A_T\psi\uppar{n+1}
\end{equation*}
as elements of $H^{-2}(\R^{3(M+n)})$.
\end{thm}
The condition $\psi\uppar{n+1} - G_T\psi\uppar{n}\in H^2(\R^{3(M+n+1)})$ means that functions in $D(H)$ satisfy the boundary condition $B\psi=M\psi$, among all functions of the form $\phi + G_T\xi$ with $\phi\in D(L)$ and $\xi\in \hilb$.


The rest of the article is devoted to the proof of Theorem~\ref{thm:sa}.
We start by discussing in detail the properties of the maps $G$ and $G_T$ as well as the domain of $H$.
 We then define the operator $A_T$ as a (distribution-valued) extension of the annihilation operator to $D(L)\oplus G_T \hilb$ and give some estimates on its regularity and growth in the particle number in Section~\ref{sect:A}. These results are put together in the proof of Theorem~\ref{thm:sa} in Section~\ref{sect:proof}. 

\section{The domain}\label{sect:dom}

Before discussing the domain of $H$, let us recall some properties of extensions of the Laplacian $L_0$,
the restriction of $L$ to vectors vanishing on the collision configurations $\mathscr{C}^n$, cf.~Eq.~\eqref{eq:DL_0}.
%
Spelling out the definition of the operator $G$, Eq.~\eqref{eq:Gdef}, we have
\begin{align*}
G\psi\uppar{n}=&  -  \left(-\frac{1}{2m}\Delta_X - \Delta_Y+n+1\right)^{-1}\frac{1}{\sqrt{n+1}}\sum_{\mu=1}^M\sum_{i=1}^{n+1} \delta_{x_\mu}(y_i)\psi\uppar{n}(X,\hat Y_i),
\end{align*}
where we recall that $X=(x_1,\dots ,x_M)\in \R^{3M}$ stands for the positions of the $x$-particles, $ Y=(y_1,\dots, y_{n+1})\in \R^{3(n+1)}$ for those of the $y$-particles, and $\hat Y_i\in \R^{3n}$ is $Y$ without the $i$-th entry.
We view $G$ both as an operator on $\hilb$ and an operator from $\hilb\uppar{n}$ to $\hilb\uppar{n+1}$, without distinguishing these by the notation.

The functions in the range of $G$ have a specific singular behaviour on the set $\mathscr{C}^{n+1}$ (defined in Eq.~\eqref{eq:Cdef})  of collision configurations.
Since
\begin{equation*}
 (-\Delta_x + \lambda^2)^{-1}\delta= \frac{\ue^{- \lambda |x|}}{4\pi |x|},
\end{equation*}
the function $G\psi\uppar{n}$ will diverge like $|x_\mu-y_i|^{-1}$ on the plane $\{x_\mu=y_i\}$. The singular set $\mathscr{C}^{n+1}$ is just the union of these planes and, due to symmetry, the divergence on each of these is exactly the same. 
To be more explicit, first observe that
the Fourier transform of $\delta_{x_\mu}(y_i)\psi\uppar{n}(X,\hat Y_i)$ equals
\begin{equation}\label{eq:fourier delta}
 \frac{1}{(2\pi)^{3/2}}\widehat \psi\uppar{n}(P+e_\mu k_i, \hat K_i),
\end{equation}
where $e_\mu$ is the inclusion of the $\mu$-th summand in $\R^{3M}=\bigoplus_{\mu=1}^M \R^3$, and $p_\mu$, $x_\mu$; $k_i$, $y_i$ are conjugate Fourier variables.  
Now consider, for simplicity, the case $M=1$. Choosing the centre of mass coordinate $s=\frac{2m}{2m+1}x+ \frac{1}{2m+1}y_{i}$ and the relative coordinate $r=\sqrt{\frac{2m}{2m+1}}(y_{i}-x)$ for the pair $x$, $y_i$ gives
\begin{align}
  &\left(-\frac{1}{2m}\Delta_x - \Delta_Y+n+1\right)^{-1}\delta_{x}(y_{i})\psi\uppar{n}(x,\hat Y_{i})\notag\\
 &=\frac{1}{(2\pi)^{3(n+1)/2+3}}\int \frac{\ue^{\ui px+\ui KY}\widehat \psi\uppar{n}(p+ k_{i}, \hat K_{i})}{\frac1{2m}p^2 + K^2+n+1}\ud p \ud K\label{eq:G fourier}\\
&=\frac{1}{(2\pi)^{3(n+1)/2}}\left(\frac{2m}{2m+1}\right)^{3/2} \frac{1}{4\pi |r|}
\int
\begin{aligned}[t]
 &\ue^{-|r|\sqrt{n+1+\frac{1}{2m+1}\xi^2 + \hat K_{i}^2}}\\
& \times\ue^{\ui \hat K_{i} \hat Y_{i}+\ui \xi s}\widehat \psi\uppar{n}(\xi, \hat K_{i})\ud \xi \ud \hat K_{i}.
 \end{aligned}\label{eq:G exp}
 \end{align}
This function diverges like $\frac{m}{2\pi(2m+1)\abs{x-y_i}}\psi\uppar{n}\left(\frac{2m}{2m+1}x+ \frac{1}{2m+1}y_{i},\hat Y_i\right)$ as $r\to 0$, i.e. $\abs{x-y_i}\to 0$.
Observe also that~\eqref{eq:G exp} is a smooth function of $s$ and $\hat Y_{i}$ for $|r|>0$ because of the  exponential decay of its Fourier transform. 

\subsection{Properties of $G$}

We will now discuss the mapping properties of the operator $G$.
As already noted in Eq.~\eqref{eq:G ker}, $G$ maps elements of $\hilb$ to the kernel of $L_0^*$.
Furthermore, $G$ is bounded from $\hilb$ to $D(L^s)$, $s<\tfrac14$:

\begin{lem}\label{lem:G bound}
 Let $0\leq s < \tfrac14$. There exists a constant $C$ such that for all $\psi\in \hilb$
 \begin{equation*}
  \norm{L^s G\psi}_\hilb \leq C \|N^{s-\frac14} \psi\|_\hilb.
 \end{equation*}
\end{lem}
\begin{proof}
 Since we will not aim for optimal estimates w.r.t.~the number $M$ of $x$-particles, we can just treat the case $M=1$ and then bound the sum over the $x$-particles by the Cauchy-Schwarz inequality.
 
 For $M=1$ and arbitrary $n\in \N$, we start from the expression~\eqref{eq:G fourier} and use the simple inequality
 \begin{align*}
 &\mathrm{Re}\left( \overline{\hat\psi\uppar{n}}(p+k_i,\hat K_i)\hat\psi\uppar{n}(p+k_j,\hat K_j)\right)\\
 &\leq \frac12 \left(\frac{|\hat\psi\uppar{n}(p+k_i,\hat K_i)|^2 k_j^2 }{k_i^2} + \frac{|\hat\psi\uppar{n}(p+k_j,\hat K_j)|^2 k_i^2 }{k_j^2}\right).
 \end{align*}
This implies (writing $L(p,K)$ for the Fourier-representation of $L$)
 \begin{align}
  &\norm{L^s G\psi\uppar{n}}^2_{\hilb\uppar{n+1}}\notag
  \leq \frac{1}{(2\pi)^3 (n+1)} \sum_{i,j=1}^{n+1} \int \frac{|\hat\psi\uppar{n}(p+k_i,\hat K_i)|^2 k_j^2 }{L(p,K)^{2-2s}k_i^2} \ud p \ud K\notag\\
  &\leq \frac{\norm{\psi\uppar{n}}^2_{\hilb\uppar{n}}}{(2\pi)^3} \sup_{(p,\hat K_{n+1})\in \R^{3(n+1)}}\sum_{j=1}^{n+1} \int \frac{k_j^2 \ud k_{n+1}}{\left(\frac1{2m}(p-k_{n+1})^2 + K^2 + n+1\right)^{2-2s} k_{n+1}^2} \notag.
  %
 \end{align}
 By splitting the summand with $j=n+1$ from the rest and changing variables $k_{n+1}\mapsto k_{n+1}/\sqrt{n+1}$, respectively $k_{n+1}\mapsto k_{n+1}/\sqrt{n+1+\hat K_{n+1}^2}$, we can further bound this by
 \begin{align}
 &\norm{L^s G\psi\uppar{n}}^2_{\hilb\uppar{n+1}} \label{eq:G bound}\\
    &\leq \frac{\norm{\psi\uppar{n}}^2_{\hilb\uppar{n}}}{(2\pi)^3}
  \left( \int \frac{ \ud k_{n+1}}{(k_{n+1}^2 + n+1)^{2-2s}}  + \sup_{\hat K_{n+1}} \hat K_{n+1}^2 \int \frac{\ud k_{n+1}}{(K^2 +n+1)^{2-2s}k_{n+1}^2} \right)\notag\\ 
  &\leq C^2 (n+1)^{2s-\frac12}\norm{\psi\uppar{n}}^2_{\hilb\uppar{n}}.
  \notag
 \end{align}
This proves the claim.
\end{proof}

As an immediate Corollary, we have that $\sum_{\mu=1}^M a(\delta_{x_\mu})$ maps $D(L)$ to $\hilb$, continuously.
\begin{cor}\label{cor:aL^-1}
 The Operator $-G^*=\sum_{\mu=1}^M a(\delta_{x_\mu})L^{-1}$ is bounded on $\hilb$.
\end{cor}

\begin{rem}[The domain of $L_0^*$]\label{rem:L_0^*}
The operator $G$ can be used to parametrise the domain of $L_0^*$.
 From the proof of Lemma~\ref{lem:G bound} one easily infers that $G:\hilb \to \ker L_0^*$ can be extended to $D(L^{-1/4})$. 
 We then have a characterisation of $D(L_0^*)$ by
 \begin{equation*}
  D(L_0^*) = D(L)\oplus G D(L^{-1/4}).
 \end{equation*}
 Since we only work on the subspace of $D(L_0^*)$ where $G$ acts on $\hilb$ we do not need this parametrisation and we will not give a detailed proof. The central argument to obtain this parametrisation is to show that the norm $\norm{G\psi}_{\hilb\uppar{n+1}}$ is equivalent to the norm on $H^{-1/2}(\R^{3(M+n)})$ by a generalisation of~\cite[Lem B.2]{Co_etal15} to an arbitrary number of particles. This implies the parametrisation above by~\cite[Prop.2.9]{BeMi14}. 
\end{rem}

As an operator on $\hilb$, $G$ has another important property:

\begin{prop}\label{prop:G}
 The operator $1-G$ on $\hilb$ has a bounded inverse, and both $1-G$ and $(1-G)^{-1}$ map the domain $D(N)$ of the number operator to itself.
\end{prop}
\begin{proof}
 We claim that 
 \begin{equation*}
  (1-G)^{-1}=\sum_{j=0}^\infty G^j
 \end{equation*}
is a bounded operator on $\hilb$. 
To prove this, we first note that, since $G$ maps $\hilb\uppar{j}$ to $\hilb\uppar{j+1}$, the sum on $\hilb\uppar{n}$ is actually finite,
\begin{equation*}
  \left((1-G)^{-1}\psi\right)\uppar{n}=\sum_{j=0}^n G^{j} \psi\uppar{n-j}.
 \end{equation*}
 The operator is thus well defined and one easily checks that it is an inverse to $1-G$.
 To show boundedness, we use Lemma~\ref{lem:G bound} to obtain

\begin{align}\label{eq:G^k est}
\norm{ G^{j} \psi\uppar{n-j}}_{\hilb\uppar{n}} \leq C^j \left({\prod_{i=0}^{j-1} (n-i)}\right)^{-\frac14}\norm{\psi\uppar{n-j}} 
\leq \frac{C^j}{(j!)^{\frac14}}  \norm{\psi\uppar{n-j}}_{\hilb\uppar{n-j}}.
\end{align}
This gives a bound on $(1-G)^{-1}$ by
\begin{equation*}
 \norm{(1-G)^{-1}\psi}_\hilb
 \leq \sum_{j=0}^\infty \norm{ G^{j} \psi}
 \leq \norm{\psi}_\hilb \left(\sum_{j=0}^\infty \frac{C^j}{(j!)^{\frac14}} \right).
\end{equation*}

The operator $1-G$ maps $D(N)$ to itself if $G$ does. From the estimate of Eq.~\eqref{eq:G bound}, we see that
\begin{equation*}
 \norm{ n (G\psi)\uppar{n}} \leq C \norm{ n^{\frac34} \psi\uppar{n-1}} \leq 2 C \norm{ (n-1) \psi\uppar{n-1}},
\end{equation*}
which proves that $G$ leaves $D(N)$ invariant.

To prove the same for $(1-G)^{-1}$, observe that $G^j$ maps $D(N)$ to itself for any finite $j$, because $G$ does. It is thus sufficient to prove the claim for $\sum_{j=j_0}^\infty G^j$, for some $j_0\in \N$. From Lemma~\ref{lem:G bound} we obtain for $j\geq 4$, as in Eq.~\eqref{eq:G^k est}, 
\begin{align*}
 \norm{ n G^j \psi\uppar{n-j}}_{\hilb\uppar{n}} 
&\leq \frac{n}{((n-3)(n-2)(n-1)n)^{\frac14}}\frac{C^j}{((j-4)!)^{\frac14}}  \norm{\psi\uppar{n-j}}_{\hilb\uppar{n-j}}\\
&\leq \tilde C  \frac{C^j}{((j-4)!)^{\frac14}}  \norm{\psi\uppar{n-j}}_{\hilb\uppar{n-j}}.
\end{align*}
This shows that $\sum_{j=4}^\infty G^j$ maps $\hilb$ to $D(N)$ and completes the proof.
\end{proof}

\subsection{Singular boundary values}\label{sect:bv}

 The asymptotic behaviour on $\mathscr{C}^{n+1}$ can be used to define local boundary operators on $D(L_0^*)$. We will define these only on the range of $G$ (acting on $\hilb$).
These boundary values will be defined by certain limits as $x_\mu-y_{n+1}\to0$. These limits exist only almost everywhere or in the sense of distributions. In the following, an expression such as $\lim_{\abs{x_\mu-y_{n+1}}\to 0} \Phi(X,Y,y_{n+1})$ should thus be read as the limit for $\eps\to 0$ of the functions $\Phi_\eps(X,Y)$ given by $\Phi(X,Y,y_{n+1})$ with $\eps=y_{n+1}-x_\mu$ fixed.
The boundary values we define can be extended to $D(L_0^*)$ (cf.~Remark~\ref{rem:L_0^*}) in the sense of distributions on the boundary without the singular set $\mathscr{C}^{n}$ (see~\cite[Lem.6]{IBCpaper} for the case of a fixed source), but we will not need this here.
Set, for suitable $\phi\in \hilb\uppar{n+1}$,
\begin{align}\label{eq:Bdef}
 \left(B\phi\right)(X,Y):
 =-\frac{2\pi (2m+1)}{m}\sqrt {n+1}\sum_{\mu=1}^M \lim_{|x_\mu-y_{n+1}|\to 0} |x_\mu-y_{n+1}|\phi(X,Y,y_{n+1}).
\end{align}
We then have:
\begin{lem}\label{lem:B}
Let $n\in \N$ and $\psi\in \hilb\uppar{n}$.
The element $G\psi\in \hilb\uppar{n+1}$ has a representative such that the limit in Equation~\eqref{eq:Bdef} exists for almost every $(X,Y)$ and the equality
 \begin{align}\notag
 \left(BG\psi\right)(X,Y) = M \psi(X, Y)\notag
 \end{align}
 holds.
\end{lem}
\begin{proof}
 The fact that  $G\psi\in \hilb\uppar{n+1}$ was proved in Lemma~\ref{lem:G bound}. It follows from the analogue of Eq.~\eqref{eq:G exp} for arbitrary $M$ that $L^{-1}\delta_{x_\nu}(y_i)\psi(X,\hat Y_i)$ has a smooth representative outside of the plane $\{x_\nu=y_i\}$. After multiplying by $|x_\mu-y_{n+1}|$, the limit as $|x_\mu-y_{n+1}|\to 0$ is thus zero a.e., except if $\mu=\nu$ and $i=n+1$. In the latter case, the limit equals $\psi(X, Y)$, in $L^2$ and thus almost everywhere, by strong continuity of the semi-group $\ue^{-t\sqrt{-\Delta + \lambda^2}}$. There is one such contribution for every $\mu$, giving the prefactor $M$ in the statement.
\end{proof}

The limit~\eqref{eq:Bdef} defining $B$ vanishes on $D(L)$, because $\psi\uppar{n}\in H^2(\R^{3(M+n)})$ has a well-defined evaluation, the Sobolev trace, on the set $\{x_\mu=y_{n+1}\}$. We can thus define an operator
\begin{equation*}
 B:D(L)\oplus G\hilb \to \hilb, \qquad \psi + G\phi\mapsto M\phi,
\end{equation*}
which can be calculated using the local expression~\eqref{eq:Bdef} outside of the singular sets.

The annihilation operator
\begin{equation*}
 a(\delta_{x_\mu})\psi\uppar{n+1}
 = \sqrt{n+1} \psi\uppar{n}(X, Y)\vert_{y_{n+1}=x_\mu},
\end{equation*}
is not defined on the singular functions in the range of $G$. It can, however, be extended by considering only the finite part at $y_{n+1}=x_\mu$, i.e. subtracting the explicit divergence before evaluation:
\begin{align}\label{eq:Adef}
&\left(A\phi\right)(X,Y) \\ 
&:=\lim_{r\to 0}\sum_{\mu=1}^M
\frac{1}{4\pi}\int\limits_{S^2}\left(\sqrt{n+1} \phi(X,Y,x_\mu+r\omega)+ \frac{m}{2\pi(2m+1)} \frac{1}{r} \frac{(B\phi)(X,Y)}{M}\right) \ud \omega.\notag
\end{align}
At least formally, $T\psi=AG\psi$ then defines an operator which preserves the number of particles. 
This operator also appears in the theory of point interactions, where it is known as the Skornyakov--Ter-Martyrosyan operator.
For fixed $n\in \N$, $T$ is composed of two contributions. First,  the finite part of the term $L^{-1}\delta_{x_\mu}(y_{n+1})\psi(X,Y)$, which actually diverges at $x_\mu=y_{n+1}$, we call this the diagonal part $T_\mathrm{d}$. The second contribution is just the evaluation of terms of the form $L^{-1}\delta_{x_\nu}(y_i)\psi(X,\hat Y_{i})$ with either $\nu\neq \mu$ or $i\neq n+1$. These diverge on different planes and are smooth on $\{x_\mu=y_{n+1}\}$ outside of the lower-dimensional set $\mathscr{C}^{n}$. We call this the off-diagonal part $T_\mathrm{od}$.
The action of $T$ can be read off from Eqs.~\eqref{eq:G fourier},\eqref{eq:G exp} and is conveniently expressed in Fourier variables. For an $n$-particle wavefunction $\psi\uppar{n}$ we have

\begin{align}
 \widehat {T_\mathrm{d} \psi\uppar{n}}(P,K)
 =&\frac{1}{4\pi}\left(\frac{2m}{2m+1}\right)^{\tfrac 32} \sum_{\mu=1}^M \sqrt{n+1 + \tfrac1{2m+1} p_\mu^2 + \tfrac1{2m}\hat P_\mu^2 + K^2} \hat \psi\uppar{n}(P,K),\label{eq:Tdiag}
 \end{align}
 which is obtained from~\eqref{eq:G exp} by developing the exponential, and
 \begin{align}
 \widehat {T_\mathrm{od} \psi\uppar{n}}(P,K)=&-\frac{1}{(2\pi)^{3}}\bigg( 
 \begin{aligned}[t]
&\sum_{\mu,\nu=1}^M \sum_{i=1}^n 
\int \frac{\hat \psi\uppar{n}(P-e_\mu \xi + e_\nu k_i,\hat K_i,\xi)}{n+1+\frac1{2m} (P-e_\mu \xi )^2 +K^2 + \xi^2} \ud \xi\\
& + \sum_{\mu\neq \nu=1}^M \int \frac{\hat \psi\uppar{n}(P-e_\mu \xi + e_\nu \xi),K)}
{n+1+\frac1{2m} (P-e_\mu \xi )^2 +K^2 + \xi^2}\ud \xi\bigg),
 \end{aligned}\label{eq:Tod}
\end{align}
which corresponds to the evaluation of~\eqref{eq:G fourier} at $y_{n+1}=x$.


It is a well-known result in the theory of point interactions (see e.g.~\cite{Co_etal15,MoSe17}) that, for a fixed number $n$ of particles,  $T=T_\mathrm{d}+T_\mathrm{od}$ is bounded from $H^{1}(\R^{3(M+n)})$ to $\hilb\uppar{n}$, and symmetric. Since, in our problem, $n$ is not fixed, the dependence of the bound on $n$ is important.
A bound on $T_\mathrm{od}$ which is independent of $n$, even though the number of terms grows linearly in $n$, was proved for the norm of $T_\mathrm{od}$ as an operator form $H^{1/2}$ to $H^{-1/2}$ by Moser and Seiringer~\cite{MoSe17} (with $M=1$). We adapt their method to prove the same for the operator from $H^1$ to $\hilb\uppar{n}$, and arbitrary $M$, in Lemma~\ref{lem:T_od}. This gives:

\begin{lem}\label{lem:T}
For all $n\in \N$  define the operator $T=T_\mathrm{d}+T_\mathrm{od}$ by~\eqref{eq:Tdiag},~\eqref{eq:Tod} 
on the domain $D(T)\uppar{n}= \hilb\uppar{n}\cap H^1(\R^{3(M+n)})$ and denote the induced operator on $\hilb$ by the same symbol.
The operator $T$ is symmetric on $D(T)=D(L^{1/2})$
and there exists a constant $C$ such that for all  $\psi \in D(L^{1/2})$
\begin{equation*}
 \norm{T\psi}_{\hilb} \leq C \|(1+L^{1/2}) \psi\|_{\hilb}.
\end{equation*}
\end{lem}
 
\begin{proof}
The statement is trivial for $T_\ud$.  The bound on $T_\mathrm{od}$ is proved in Lemma~\ref{lem:T_od}.
To show symmetry of $T_\mathrm{od}$ on its domain, one may use the representation
\begin{align}\label{eq:T sym}
 T_\mathrm{od} \psi\uppar{n}(X,\hat Y_{n+1}) = -\sum_{(\mu,n+1)\neq (\nu, i)} \tau_{x_\mu}(y_{n+1})L^{-1}\delta_{x_\nu}(y_i) \psi \uppar{n}(X,\hat Y_i),
\end{align}
where $\tau_{x_\mu}(y_{n+1})$ denotes evaluation at $y_{n+1} = x_\mu$ (outside of $\mathscr{C}^{n}$). This proves the claim, because of the bounds on $T_\mathrm{od}$ obtained before and because  $\tau^*_x=\delta_x$, as a map from smooth functions to distributions.
\end{proof}

For the boundary operator $A$ we can now prove:

\begin{lem}\label{lem:A}
 Let $\psi\in \hilb\uppar{n}$.
 The limit in Equation~\eqref{eq:Adef} with $\phi=G\psi$ exists in $H^{-1}(\R^{3(M+n)})$ and $AG\psi=T\psi$.
\end{lem}
\begin{proof}
 It is sufficient to prove the claim for a fixed $\mu\in \{1,\dots, M\}$. We start with the diagonal part $T_\mathrm{d}$, which is clearly an operator from $\hilb\uppar{n}$ to $H^{-1}(\R^{3(M+n)})$. Consider the representation of $G$ in Eq.~\eqref{eq:G exp}. This shows, for $M=1$, that the limit for $|x-y_{n+1}|\to 0$ of
 \begin{equation*}
 \left(L^{-1}\delta_x(y_{n+1})\psi\right)(x,Y) - \frac{m}{2\pi(m+1)} \frac{1}{|x-y_{n+1}|} \psi(\tfrac{2m}{2m+1}x+\tfrac{1}{2m+1}y_{n+1},Y)
 \end{equation*}
exists in $H^{-1}$ for $\psi\in L^2$, and equals $-(2\pi)^3T_\mathrm{d}\psi$ because this limit is (up to a prefactor) just the derivative at $r=0$ of the semi-group $\ue^{-r T_\mathrm{d}}$. To complete the proof for $M=1$,
we need to show that the error made by replacing $\psi(x,Y)$ with $\psi(\tfrac{2m}{2m+1}x+\tfrac{1}{2m+1}y_{n+1},Y)$ converges to zero as $r\to0$. This follows, by duality, from the fact that for $f\in H^1(\R^{3(1+n)})$
\begin{align}
 \lim_{r\to 0}&\,\frac{1}{r}\int\limits_{S^2} \left(f(x-\tfrac{1}{2m+1}r\omega,Y)- f(x,Y)\right) \ud \omega
 = -\frac{1}{2m+1}\int\limits_{S^2} \omega\cdot \nabla f(x,Y) \ud \omega=0\label{eq:av 0}
\end{align}
in $\hilb\uppar{n}$. The generalisation to arbitrary $M$ is straightforward, completing the argument for $T_\ud$.

The off-diagonal part $T_\mathrm{od}$ is a bounded operator from $\hilb\uppar{n}$ to $H^{-1}(\R^{3(M+n)})$ by Lemma~\ref{lem:T} and duality. The estimates of Lemma~\ref{lem:T_od} also yield a uniform bound for the operator obtained by evaluation at $y_{n+1}-x_\mu=\eps$, whose  integral kernel differs from that of $T_\mathrm{od}$ by a factor $\ue^{\ui \eps \xi}$. This implies dominated convergence for $\eps\to 0$.
\end{proof}
We have thus defined a second boundary value operator 
\begin{equation*}
 A:D(L)\oplus G\hilb \to D(L^{-1/2}), \qquad \psi + G\phi\mapsto \sum_{\mu=1}^M a(\delta_{x_\mu})\psi + T\phi,
\end{equation*}
where $T$ is extended to $T:\hilb\uppar{n}\to H^{-1}(\R^{3(M+n)})$ by duality.
This is an extension of the annihilation operator, originally defined on $D(L)$. There are of course many such extensions, e.g. the map $G\phi\mapsto 0$ provides an example. The extension $A$ is special in that it is also a sum of local operators, in the sense that $(A\psi)\uppar{n}(X,Y)$ is a sum over $\mu$ in which each term is completely determined by $\psi$ restricted to any neighbourhood of the point $(X,Y,x_\mu)\in \mathscr{C}^{n+1}$.

\begin{rem}[The model with one $y$-particle]\label{rem:N=1}
 The results we have obtained so far are sufficient to discuss the model with at most one $y$-particle. For the cases $M=1,2$ this is essentially the model introduced in~\cite{thomas1984}. Certain sectors of the Galilean Lee model~\cite{levy1967, schrader1968} can also be described in a very similar way.
 
 Consider the subspace $D$ of $\hilb\uppar{0}\oplus\hilb\uppar{1}$ formed by elements $\psi=(\psi\uppar{0},\psi\uppar{1})$ with $\psi\uppar{0}\in D(L)=H^2(\R^{3M})$, $\psi\uppar{1}\in D(L)\oplus G\hilb\uppar{0}$. 
 On  this space, both $B\psi\uppar{1}$ and $A\psi\uppar{1}$ are well defined. Since $BG\phi=M\phi$ we also have 
 \begin{equation*}
 \psi\uppar{1}- GB\psi\uppar{1}/M\in D(L)=H^2(\R^{3(M+1)}) .
 \end{equation*}
 Using first that $L_0^*G=0$ and then the symmetry of $(L,D(L))$,
 we find the identity for $\psi, \phi\in D$ 
 \begin{align*}
  &\langle L_0^* \psi\uppar{1},\phi\uppar{1} \rangle -  \langle  \psi\uppar{1},L_0^*\phi\uppar{1} \rangle\\
  &=\langle L_0^* (\psi\uppar{1}-GB\psi\uppar{1}/M),\phi\uppar{1} \rangle -  \langle  \psi\uppar{1},L_0^*(\phi\uppar{1}-GB\phi\uppar{1}/M) \rangle\\
  %
  &=\langle L (\psi\uppar{1}-GB\psi\uppar{1}/M),GB\phi\uppar{1}/M \rangle
  -  \langle GB \psi\uppar{1}/M,L(\phi\uppar{1}-GB\phi\uppar{1}/M)\rangle.
  %
\end{align*}
Since $LG=-\sum_{\mu=1}^M a^*(\delta_{x_\mu})$ this equals
\begin{align*}
 \frac{1}{M}\sum_{\mu=1}^M \bigg(-&\langle a(\delta_{x_\mu})(\psi\uppar{1}-GB\psi\uppar{1}/M), B\phi\uppar{1} \rangle\\
 &+  \langle B \psi\uppar{1},a(\delta_{x_\mu})(\phi\uppar{1}-GB\phi\uppar{1}/M) \rangle\bigg).
\end{align*}
If $B\psi\uppar{1}$ and $B\phi\uppar{1}$ are elements of the domain $D(T)\uppar{0}= H^1(\R^{3M})$  of the symmetric operator $T$ we can add the term $(\langle T B\psi\uppar{1}, B\phi\uppar{1}\rangle - \langle  B\psi\uppar{1},T B\phi\uppar{1}\rangle)/M^2=0$ to this equation. Since $AG=T$ it then becomes
\begin{equation}
 \langle L_0^* \psi\uppar{1},\phi\uppar{1} \rangle -  \langle  \psi\uppar{1},L_0^*\phi\uppar{1} \rangle
 =\langle B\psi\uppar{1}/M,A\phi\uppar{1}\rangle -\langle A\psi\uppar{1},B\phi\uppar{1}/M\rangle.
\end{equation}
This implies that the operator 
\begin{equation*}
H\uppar{1}(\psi\uppar{0},\psi\uppar{1}) = (L\psi\uppar{0} + A\psi\uppar{1}, L_0^* \psi\uppar{1}) 
\end{equation*}
 is symmetric if we impose the boundary condition $B\psi\uppar{1}=M\psi\uppar{0}$, i.e.~on the domain
 \begin{equation*}
  D(H\uppar{1})=\Big\{ \psi \in \hilb\uppar{0}\oplus\hilb\uppar{1} \Big\vert \psi\uppar{0}\in D(L),\, \psi\uppar{1}\in D(L)\oplus G\hilb\uppar{0},\, B\psi\uppar{1}=M\psi\uppar{0}\Big\}.
 \end{equation*}

 One can prove that $H\uppar{1}$ is self-adjoint, for example by constructing its resolvent along the lines of~\cite{levy1967, thomas1984}, or by adapting our proof in Section~\ref{sect:proof}. 
For $\psi\in D(H\uppar{1})$ we have
\begin{equation}\label{eq:L+a^*}
L_0^* \psi\uppar{1}=L(\psi\uppar{1}-G\psi\uppar{0})=L\psi\uppar{1} + \sum_{\mu=1}^M a^*(\delta_{x_\mu})\psi\uppar{0},
\end{equation}
where the right hand side is a sum in $H^{-2}(\R^{3(M+1)})$. We can thus also write 
\begin{equation*}
H\uppar{1}\psi=L\psi + A\psi\uppar{1}+ \sum_{\mu=1}^M a^*(\delta_{x_\mu})\psi\uppar{0}.
\end{equation*}

It is important to note that we have used the fact that $\psi\uppar{0}\in H^2(\R^{3M})\subset D(T)\uppar{0}$. This does not carry over to cases with more $y$-particles, since  $B\psi\uppar{1}=\psi\uppar{0}$ implies that $\psi\uppar{1}$ is \emph{not} in $H^1(\R^{3(M+1)})$ (or even $H^{1/2}(\R^{3(M+1)})$), if $\psi\uppar{0}\neq 0$.
\end{rem}
\subsection{The Operator $K$ and its extension}\label{sect:K}
Up to now we have developed the theory very much in parallel to~\cite{LaSch18}, where we treated in particular the two-dimensional variant of our model. There, we proved that $L_0^*+A$  is self-adjoint of the domain $D(L)\oplus G\hilb$ with the boundary condition $B\psi=M\psi$. For the three-dimensional model, this cannot be true as such, since $AG=T$ is defined on $D(T)$ but if $B\psi\uppar{n}=\psi\uppar{n-1}\neq 0$, then $\psi\uppar{n}\notin H^1(\R^{3(M+n)})\cap\hilb\uppar{n}=D(T)\uppar{n}$, since such a function must diverge like $\psi\uppar{n-1}(X, \hat Y_n)|x_\mu-y_n|^{-1}$ as $y_n\to x_\mu$. Consequently, we should not expect $L_0^*+A$ to map this domain to $\hilb$. This problem cannot be remedied by simply interpreting the operators as quadratic forms, since one can also show that $G\hilb\uppar{n-1}\cap H^{1/2}   (\R^{3(M+n)}) =\{0\}$ (see ~\cite[Prop.4.2]{LaSch18}).

\begin{rem}\label{rem:cancellations}
 Of course, the fact that $G\psi$ is not in the domain $D(T)=D(L^{1/2})$ of $T$ does not immediately imply that $TG\psi\notin \hilb$, but only that $T_\ud G\psi\notin \hilb$ and $T_\mathrm{od}G\psi \notin \hilb$, separately.
 Cancellations between these two terms can make $TG$ well defined on $D(L)\oplus G\hilb$. In this case, one can proceed with the operator $L_0^*+A$. This happens in the variant of our model with fixed $x$-particles~\cite{IBCpaper}, which formally corresponds to taking $m=\infty$ (see Eq.~\eqref{eq:alpha_m} below). Similar cancellations occur in the Galilean Lee model~\cite{levy1967, schrader1968}, when formulated in our language, due to the constraint on the masses of the different particles. 
 Proposition~\ref{prop:A_T} below shows that, in our model, there are no such cancellations.
\end{rem}

Our solution to the problem of defining $AG$ is to change the regularity of  elements of the domain in such a way that the singularities of $A$ and $L_0^*$ cancel, similar to the cancellation in $L+a^*$ on $D(H\uppar{1})\subset D(L_0^*)$ in the case of one $y$-particle, see Eq.~\eqref{eq:L+a^*}.  
Let $K$ be the operator
\begin{equation*}
K=L+T 
\end{equation*}
with $D(K)=D(L)$ and $T$ given by~\eqref{eq:Tdiag},~\eqref{eq:Tod}. By Lemma~\ref{lem:T} and the Kato-Rellich theorem, $K$ is self-adjoint and bounded from below. Let $K_0$ be the restriction of $K$ to $D(L_0)$, the vectors $\psi\in D(L)$ with $\psi\uppar{n}\in H^2_0(\R^{3(M+n)}\setminus\mathscr{C}^n)$.  In analogy with $G$, we define a map
\begin{equation}\label{eq:GTdef}
 G_T:\hilb \to \ker (K_0^* + c_0), \qquad  G_T\psi:=  - \sum_{\mu=1}^M (K+c_0)^{-1} a^*(\delta_{x_\mu}) \psi,
\end{equation}
where $c_0>-\min \sigma(K)$ is a fixed constant.  
The important properties of $G$ carry over to $G_T$ by perturbation theory.
\begin{prop}\label{prop:G_T}
 Let $c_0>-\min \sigma(K)$ and $G_T$ be defined by~\eqref{eq:GTdef}. For $0\leq s < \tfrac14$ there exists a constant $C$ such that 
 \begin{equation*}
  \norm{L^s G_T\psi}_\hilb \leq C \|N^{s-\frac14} \psi\|_\hilb.
 \end{equation*}
Moreover, $1-G_T$ has a continuous inverse on $\hilb$ and $1-G_T$ as well as $(1-G_T)^{-1}$ map the domain of the number operator $D(N)$ to itself.
\end{prop}
\begin{proof}
 By the resolvent formula we have
 \begin{equation}\label{eq:GT pert}
  G_T\psi = G\psi - L^{-1}(T+c_0) G_T \psi.
 \end{equation}
 The bound on $L^sG_T$ then follows from the bound on $L^s G$, Lemma~\ref{lem:G bound}, and the fact that $L^{-1/2}T$ is bounded on $\hilb$, by Lemma~\ref{lem:T}.
 The continuity and the mapping properties  of 
 \begin{equation*}
  (1-G_T)^{-1}=\sum_{j=0}^\infty G^j_T
 \end{equation*}
 follow from this by exactly the same proof as in Proposition~\ref{prop:G}.
\end{proof}

This proposition implies that $(1-G_T)\psi\in D(N)$ if and only if $\psi \in D(N)$, so we have
 \begin{align} 
 D(H)&= \{ \psi \in D(N): (1-G_T)\psi \in D(L) \}\notag\\
 &= \{ \psi \in \hilb: (1-G_T)\psi \in D(L) \}\notag\\
 &=   (1-G_T)^{-1} D(L).\label{eq:1-G D(H)}
 \end{align}
 This also shows that $D(H)$ is dense in $\hilb$, by continuity  and surjectivity of $(1-G_T)^{-1}$.
Concerning the regularity of vectors in the range of $G_T$, applying the identity~\eqref{eq:GT pert} twice we find for any $n\geq 1$
\begin{equation}\label{eq:G_T decomp}
 G_T\psi\uppar{n}=G\psi \uppar{n} - L^{-1}(T+c_0) G \psi\uppar{n} + (L^{-1}(T+c_0))^2 G_T \psi\uppar{n}.
\end{equation}
The first term is an element of $H^{s}(\R^{3(M+n+1)})$, $s<\frac12$, by Lemma~\ref{lem:G bound} and diverges like $|x_\mu-y_i|^{-1}$ near $\mathscr{C}^{n+1}$ by Lemma~\ref{lem:B}.
The operator $TL^{-1}T$ is bounded on $\hilb\uppar{n}$ by Lemma~\ref{lem:T}, so the last term above is an element of $H^2(\R^{3(M+n+1)})$. It can thus be evaluated on the co-dimension-three set $\mathscr{C}^{n+1}$ of collision configurations. The term  $L^{-1}(T+c_0) G \psi\uppar{n}$ is in $H^1(\R^{3(M+n+1)})$ (even in $H^s$, $s<3/2$), so it should have a less pronounced divergence on $\mathscr{C}^{n+1}$ than the $|x_\mu-y_i|^{-1}$-divergence of $G\psi \uppar{n}$. In fact, we will show in Proposition~\ref{prop:A_T} that this term diverges logarithmically. We can thus define $BG_T \psi=M\psi$ by the same limit~\eqref{eq:Bdef} as for $G$.
We will define a modification $A_T$ of the operator $A$ on the range of $G_T$ in the next section, see Proposition~\ref{prop:A_T}.

\section{Extension of the annihilation operator}\label{sect:A}

In this section we will analyse the divergence of $G_T\psi$ on the sets $\mathscr{C}^n$ in order to define a local boundary operator $A_T$ that extends $\sum_{\mu=1}^M a(\delta_{x_\mu})$ to $D(L)\oplus G_T\hilb$.
%
One should think of functions in the range of $G_T$ as having an expansion of the form
\begin{equation*}
  (G_T\psi)\uppar{n+1}(X,Y) \sim \psi\uppar{n}(X,\hat Y_{n+1}) \left( \frac{b_1}{|x_\mu-y_{n+1}|} + b_2 \log|x_\mu-y_{n+1}| \right) + F(X,Y)
\end{equation*}
near $x_\mu=y_{n+1}$.
Here, the constant $b_1$ comes from the expansion of $G$ given in Eq.~\eqref{eq:G exp}, $b_2$ is determined by the second term in Eq.~\eqref{eq:G_T decomp}, and $F(X,Y)$ has an appropriate limit as $\abs{x_\mu-y_{n+1}}\to 0$. As in Section~\ref{sect:bv}, we would then define boundary value operators $B$, $A_T$ such that any $\phi\in D(L)\oplus G_T \hilb$ has the expansion for $\abs{x_\mu-y_{n+1}}\to 0$
\begin{equation*}
 \phi\uppar{n+1}(X,Y)\sim \frac{(B\phi\uppar{n+1})(X,\hat Y_{n+1})}{M} f(\abs{x_\mu-y_{n+1}}) + (A_T \phi\uppar{n+1})(X,\hat Y_{n+1}) + o(1),
\end{equation*}
where $f(\abs{x_\mu-y_{n+1}})$ is the divergent function in the expansion above. We will justify this intuition by defining the operators $B, A_T$ and showing that they are given by appropriate limits, in the sense of distributions.

We define the the map
\begin{align*}
 B: D(L)\oplus G_T \hilb \to \hilb,\qquad \phi + G_T\psi \mapsto M \psi.
\end{align*}
As a consequence of Proposition~\ref{prop:A_T} below, $B$ is a sum of local boundary operators given by the same expression, Eq.~\eqref{eq:Bdef} (where the limit is taken in $H^{-1}$), as the corresponding operator on $ D(L)\oplus G\hilb$. To define $A_T$, we set for appropriate $\phi\in \hilb\uppar{n+1}$

\begin{align}\label{eq:A_Tdef}
\big(A_T\phi&\big)(X,Y)\\
&:=
\lim_{r\to 0}\sum_{\mu=1}^M
\frac{1}{4\pi}\int\limits_{S^2}
\left(\sqrt{n+1} \phi(X,Y,x_\mu+r\omega)+ f_m(r) (B\phi)(X,Y)\right) \ud \omega,\notag
\end{align}
where $r>0$,
\begin{equation*}
 f_m(r)=\frac1M
\left(\frac{m}{2\pi(2m+1)} \frac{1}{r} + \gamma_m\log(r)\right), 
\end{equation*}
and $\gamma_m$ is the constant
 \begin{equation}\label{eq:alpha}
   \gamma_m=\frac{1}{(2\pi)^3}\left(\frac{2m}{2m+1}\right)^3
 \left(\frac{2\sqrt{m(m+1)}}{2m+1} - (2m+1) \tan^{-1}\left(\frac{1}{2\sqrt{m(m+1)}}\right)\right).
 \end{equation}
Note that for $\phi\in H^2(\R^{3(M+n+1)})$, $A_T\phi$ equals the usual annihilation operator.
We will show that the formula for $A_T$ defines a map
\begin{equation}\label{eq:A_T dom}
 A_T:D(L)\oplus G_T \hilb \to D(L^{-1/2})\,,\qquad \phi+G_T\psi \mapsto  \sum_{\mu=1}^M a(\delta_{x_\mu}) \phi + T\psi +S\psi,
\end{equation}
where $T=AG$ as before and $S:D(L^\eps)\to \hilb$ is symmetric, for any $\eps>0$.

Observe also that 
\begin{equation}\label{eq:alpha_m}
\lim_{m\to \infty} \gamma_m=\frac{1}{(2\pi)^3}\left(1-(\tan^{-1})'(0)\right)=0.
\end{equation}
which explains why there is no logarithmically divergent term in the case of fixed $x$-particles, treated in~\cite{IBCpaper}.
The result of this section is:
\begin{prop}\label{prop:A_T}
 Let  $\psi\in \hilb\uppar{n}$. Then for $\phi=G_T\psi$ the limit in Eq.~\eqref{eq:A_Tdef} exists in $H^{-1}(\R^{3(M+n)})$ and $S\psi:=(A_TG_T-T)\psi$ defines a symmetric operator on
 $D(S)=\hilb\uppar{n} \cap H^\eps(\R^{3(M+n)})$, for any $\eps>0$.
\end{prop}
We will give an outline of the proof here and provide some of the more technical points as separate lemmas in the appendix.
\begin{proof}
 Let $R\psi:=- L^{-1}T G \psi$. Then, in view of Eq.~\eqref{eq:G_T decomp}, we have
 \begin{equation}\label{eq:G_T R}
  G_T\psi = G\psi + R\psi + \left((L^{-1}(T+c_0))^2G_T -c_0 L^{-1} G\right) \psi.
 \end{equation}
The sum of $G\psi$ and the $1/r$-term in $f_m(r)$ converges to $AG\psi=T\psi \in H^{-1}$ by Lemma~\ref{lem:A}. Since the last term in Eq.~\eqref{eq:G_T R} is an element of $H^2(\R^{3(M+n+1)})$, it has a Sobolev trace on $\{x_\mu=y_{n+1}\}$ and the usual annihilation operator is well defined on this term. We denote this evaluation by
\begin{equation}\label{eq:S reg}
S_\mathrm{reg}\psi:= \sum_{\mu=1}^M a(\delta_{x_\mu}) \left((L^{-1}(T+c_0))^2 G_T -c_0 L^{-1} G\right) \psi.
\end{equation}
It then remains to show the convergence of the sum of $R\psi$ and the logarithmic term in $f_m(r)$. 
It is sufficient to prove convergence in $\hilb\uppar{n}$ for $\psi\in H^\eps(\R^{3(M+n)})$, $\eps>0$,  convergence in $H^{-\eps}$ for $\psi \in \hilb\uppar{n}$ then follows by duality.

We will focus on the calculation of the asymptotic behaviour at $r=0$ in the case $M=1$, $n=0$ here, the full argument is provided in Lemma~\ref{lem:S exists}.
We set $R_\mathrm{d}=- L^{-1}T_\mathrm{d} G$, $R_\mathrm{od}=- L^{-1}T_\mathrm{od} G$ and start by analysing $(R_\mathrm{d}\psi)(x,y)$ at $x=y$. By Eqs.~\eqref{eq:G fourier},~\eqref{eq:Tdiag} we have, with a change of variables $\sigma=p+k$, $\rho=k -\frac{1}{(2m+1)} \sigma$,
\begin{align}
 &(R_\mathrm{d}\psi)(x,y)\notag\\
   &=\frac{1}{(2\pi)^{9/2}}\frac{1}{4\pi}\left(\frac{2m}{2m+1}\right)^{3/2}
   \int \frac{\ue^{\ui  px+ \ui ky}\sqrt{1+\frac{1}{2m+1}p+ k^2} }{L(p,k)^2}\hat\psi(p + k) \ud p \ud k\notag\\
 &=\begin{aligned}[t]
&\frac{1}{2(2\pi)^{4}}\left(\frac{2m}{2m+1}\right)^{3/2} \\
&\times\frac{1}{(2\pi)^{3/2}} 
\int \frac{\ue^{\ui \sigma s + \ui \rho r}\sqrt{1+\frac{2m+2}{2m+1}\rho^2 + b_1 \sigma^2 + b_2 \rho\sigma} }{(1+\frac{1}{2m+1} \sigma^2 +  \frac{2m+1}{2m}\rho^2 )^2}\hat\psi(\sigma) \ud \sigma \ud \rho,
\end{aligned}
\label{eq:R_d}
\end{align}
where $s$, $r$ are the centre of mass and relative coordinate and $b_1=\tfrac{4m^2+2m+1}{(2m+1)^3}$, $b_2=\tfrac{2}{(2m+1)^2}$. This acts on $\psi$ as a Fourier multiplier with the function given by the $\rho$-integral. The singularity of this integral depends only on the behaviour of the integrand at infinity, so we replace the square root in the numerator by $\sqrt{(2m+2)/(2m+1)}|\rho|$. The error we make by this replacement is integrable in $\rho$ and the evaluation at $x=y$ gives rise to a Fourier multiplier with a function that grows no faster than $|\sigma|^\eps$ (see Lemma~\ref{lem:S bound}).
We thus have to calculate the asymptotic behaviour as $r\to 0$ of
\begin{equation}\label{eq:asym T^d}
 \frac{1}{2(2\pi)^4}\left(\frac{2m}{2m+1}\right)^{3/2}\sqrt{\frac{2m+2}{2m+1}} \int \frac{\ue^{\ui \rho r}|\rho|}{(1 +\frac1{2m+1}\sigma^2+\frac{2m+1}{2m} \rho^2)^2}\ud \rho.
\end{equation}
The integral
\begin{equation*}
 \int_{\R^3} \frac{\ue^{\ui \lambda \xi}|\xi|}{(1+\xi^2)^2} \ud \xi = 4\pi \int_0^\infty  \frac{\sin(t)t^2}{(\lambda^2+t^2)^2} \ud t= 2\pi \int_0^\infty \frac{\sin(t)+t\cos(t)}{\lambda^2+t^2} \ud t
\end{equation*}
has an expansion given by $-4\pi \log(\lambda) + \mathcal{O}(1)$ as $\lambda\to 0$. We thus find that the expression~\eqref{eq:asym T^d} behaves like
\begin{equation*}
 -\frac{1}{(2\pi)^3}\left(\frac{2m}{2m+1}\right)^{7/2} \sqrt{\frac{2m+2}{2m+1}}\log\left(r\sqrt{\tfrac{2m}{2m+1}(1+\tfrac{1}{2m+1}\sigma^2)}\right),
\end{equation*}
up to remainders that are uniformly bounded in $\sigma$ as $r\to 0$.

We now turn to $R_\mathrm{od}=L^{-1}T_\mathrm{od} L^{-1}\delta_x\psi$. We have (cf.~\eqref{eq:Tod} and note that $T_\mathrm{od}$ here is the operator on the one-particle space)
\begin{align*}
 &\left(\widehat {T_\mathrm{od}L^{-1}\delta_x\psi}\right)(p,k)\\
 &=-\frac{1}{(2\pi)^{9/2}} \int \frac{1}{2+\frac{1}{2m}(p-\xi)^2+ \xi^2+k^2}\frac{\hat\psi(p+k)}{1+\frac1{2m}(p+k-\xi)^2+\xi^2} \ud \xi.
\end{align*}
Using the same variables $\rho, \sigma, s,r$ as before, this gives
\begin{align}
 &\left(R_\mathrm{od}\psi\right)(s,r)\label{eq:R od1}\\
 &= \begin{aligned}[t]-\frac{1}{(2\pi)^{6+3/2}}
 \int& \frac{\ue^{\ui \sigma s}\ue^{\ui \rho r}}
 {(1+\tfrac1{2m+1}\sigma^2+\tfrac{2m+1}{2m}\rho^2)(1 + \frac{1}{2m}(\sigma-\xi)^2 +\xi^2)}\\
& \frac{\widehat\psi(\sigma)}{2+\xi^2 + \tfrac1{2m+1}(\sigma-\xi)^2+\tfrac{2m+1}{2m}(\rho+\tfrac1{2m+1}\xi)^2} \ud \xi \ud\sigma \ud\rho.
 \end{aligned}\notag
\end{align}
Similar to the case of $R_\mathrm{d}$, this acts as a Fourier multiplier by the function given by the integral over $\xi$ and $\rho$. To simplify the calculation of this integral, we replace $(\sigma-\xi)^2$ in the denominator by $\sigma^2+\xi^2$ (the error again has better decay in $\xi$ and $\rho$, see Eq.~\eqref{eq:tau diff}).
 For the expression resulting from the denominator of the last line we then gather the terms
\begin{align*}
 \xi^2+ \tfrac1{2m+1} \xi^2 + \tfrac{2m+1}{2m}(\rho+\tfrac1{2m+1}\xi)^2
 &=\tfrac{2m+1}{2m}(\xi+\tfrac1{2m+1}\rho)^2 + \tfrac{2m+2}{2m+1} \rho^2,
\end{align*}
making apparent that the $\xi$-integral is now a convolution. This can be evaluated using the Fourier transform
\begin{align*}
 \int_{\R^3}  \frac{e^{-\ui kx} e^{-\lambda |x|}}{|x|^2}\ud x
  =4\pi \int_0^\infty  \frac{\sin(|k|t)}{|k|t} e^{-\lambda t}\ud t
 =\frac{4\pi}{|k|} \tan^{-1}\left(\frac{|k|}{\lambda}\right).
\end{align*}
The result is
\begin{align}
 &\int \frac{1}{(\beta+\frac{2m+1}{2m}\xi^2)(\gamma+\frac{2m+1}{2m}(\xi+\frac1{2m+1}\rho)^2)}\ud \xi\notag\\
 &=2\pi^2 \left(\frac{2m}{2m+1}\right)^2 \frac{2m+1}{|\rho|}\tan^{-1}\left(\frac{|\rho|}{\sqrt{2m(2m+1)}(\sqrt \beta +\sqrt \gamma)}\right),
\end{align}
where $\beta=1+\tfrac1{2m}\sigma^2$, $\gamma=2+\tfrac1{2m+1}\sigma^2 + \tfrac{2m+2}{2m+1} \rho^2$.
From this we can see that the divergence of $R_\mathrm{od}\psi$ stems from the insufficient (cubic) decay of the integrand for large $|\rho|$, as for $R_\ud$. For the analysis of this divergence, we can thus replace the $\tan^{-1}$ by its limit as $\rho\to\infty$, which equals (note the $\rho$-dependence of $\gamma$)
 \begin{equation*}
 \lim_{\rho \to \infty}\tan^{-1}\left(\frac{|\rho|}{\sqrt{2m(2m+1)}(\sqrt \beta +\sqrt \gamma)}\right) = \tan^{-1}\left(\frac{1}{2\sqrt{m(m+1)}}\right).
\end{equation*}
The asymptotics of the remaining $\ud \rho$-integral can be evaluated as for $R_\ud$. One finds that $(R_\mathrm{od}\psi)(s,r)$ has the asympotic behaviour
\begin{equation*}
 \frac{1}{(2\pi)^3} \left(\frac{2m}{2m+1}\right)^2 2m  \tan^{-1}\left(\frac{1}{2\sqrt{m(m+1)}}\right)\log (r)\psi(s),
\end{equation*}
with a convergent remainder, as $r\to0$. Bounds on the convergent part as an operator on $H^\eps$ are provided in Lemma~\ref{lem:S bound}.
Consequently,
\begin{equation*}
 (R\psi)(s,r)= -\gamma_m\psi(s)\log(r) + \mathcal{O}(1),
\end{equation*}
with a remainder that converges in $\hilb\uppar{n}$ as $r\to 0$.

This remains true with $\psi(s)$ replaced by $\psi(x)=\psi(s+\tfrac1{2m+1} r)$ after averaging $\omega=(y-x)/\abs{y-x}$, as in Eq.~\eqref{eq:av 0}.
This completes the proof for $M=1$, $n=0$.

For arbitrary $M$ and $n$, the key observation
is that, although $R$ is given in terms of sums corresponding to different combinations of creating and annihilating a particle on the planes $\{x_\nu=y_i\}$, $\nu\in \{1,\dots M\}$, $i\in \{1,\dots, n\}$ only some of the contributions are actually singular. These behave in a similar way as for $M=1$, $n=0$, see Lemma~\ref{lem:S exists} for details.

We then have
\begin{equation*}
 A_T G_T\psi = T\psi +S\psi.
\end{equation*}
Similarly to $T$, the operator $S$ is a sum of real Fourier multipliers (in this case of logarithmic growth) and integral operators. Its symmetry on $D(S)$ is shown in Lemma~\ref{lem:S sym}.
\end{proof}

\section{Proof of Theorem~\ref{thm:sa}}\label{sect:proof}

We will now prove the self-adjointness of the operator 
\begin{equation*}
H =(1-G_T)^*(K+c_0) (1-G_T) -c_0  + S 
\end{equation*}
 on the domain 
\begin{align*}
 D(H)=\left\{ \psi \in \mathscr{H} \Big\vert  \psi - G_T\psi\in D(L) \right\}.
\end{align*}
Equality of this domain and the one given in Theorem~\ref{thm:sa} was shown in Eq.~\eqref{eq:1-G D(H)}.
The remaining statements of the theorem follow from~\eqref{eq:H trafo} and~\eqref{eq:H create} in view of the results of Section~\ref{sect:dom}.

\begin{lem}
 The operator $H_0=(1-G_T)^*K(1-G_T)$ is self-adjoint on $D(H)$ and bounded from below.
\end{lem}
\begin{proof}
This follows directly from the invertibility of $(1-G_T)$.
\end{proof}

\begin{lem}
 Let  $(S,D(S))$ be the symmetric operator on $\hilb\uppar{n}$ defined in Proposition~\ref{prop:A_T} and denote its extension to $\hilb$ by the same symbol. Then $(S,D(S))$ is infinitesimally $H_0$-bounded.
\end{lem}
\begin{proof}
 We decompose $S\psi=SG_T\psi + S(1-G_T)\psi$ and estimate both terms separately. For the action of $S$ on the singular part, $G_T\psi$, we have by Lemma~\ref{lem:S bound} and Proposition~\ref{prop:G_T}
 \begin{align*}
  \norm{SG_T \psi}_\hilb \leq C_\eps \norm{L^\eps G_T \psi}_\hilb \leq C'_\eps \norm{\psi}_\hilb,
 \end{align*}
for any $0<\eps<\tfrac14$. On the regular part we have, again by Lemma~\ref{lem:S bound},
\begin{align*}
  \norm{S(1-G_T) \psi}_\hilb 
  &\leq C_\eps \norm{L^\eps (1-G_T) \psi}_\hilb\\
  &\leq \delta \norm{ L (1-G_T) \psi}_\hilb + C_{\delta} \norm{(1-G_T) \psi}\\
  &\leq \delta \|(1-G_T)^{-1} \| \norm{H_0 \psi}_\hilb + C_\delta  \norm{(1-G_T) }\norm{\psi}_\hilb,
\end{align*}
for arbitrary $\delta>0$. This proves the claim.
\end{proof}

In view of Eq.~\eqref{eq:H trafo} this proves that $H$ is self-adjoint on $D(H)$, by the Kato-Rellich theorem. It is also immediate that $H$ is bounded from below.

\begin{rem}\label{rem:Dressing}
The transformation $(1-G_T)$ achieves something similar to the Gross transformation in the Nelson model, in that it provides a transformation relating the interacting operator and a perturbation of the free operator. 
 Applying the inverse of $(1-G_T)$ to Eq.~\eqref{eq:H trafo} we have explicitly
 \begin{equation*}
  (1-G_T^*)^{-1}H (1-G_T)^{-1}=L+T+c_0 - c_0 (1-G_T^*)^{-1}(1-G_T)^{-1}+(1-G_T^*)^{-1}S(1-G_T)^{-1},
 \end{equation*}
 and our proof shows that this is indeed s self-adjoit operator on $D(L)$.
The key difference is that the Gross transformation is a Weyl operator, and thus constructed starting from a one-particle function, while $G_T$ contains in $T$ an $n$-body interaction that cannot be expressed in such a way.

Applying the corresponding transformation to the model with an unltraviolet cutoff $\Lambda>0$ allows for a reformulation of our result in the language of renormalisation, as in~\cite{LaSch18}.
More precisely, one introduces the cutoff interaction $v_\Lambda$ and the associated objects (here for $M=1$)
\begin{align*}
 G_\Lambda &= - L^{-1} a^*(v_\Lambda),\\
 T_\Lambda &= - a(v_\Lambda)L^{-1} a^*(v_\Lambda)+ \frac{m\Lambda}{\pi^2 (2m+1)},  \\
 K_\Lambda &= L + T_\Lambda,\\
 G_{T_\Lambda} &=- (K_\Lambda + c_0)^{-1}a^*(v_\Lambda).
\end{align*}
Then, by the same algebra that leads to Eqs.~\eqref{eq:H create},\eqref{eq:H trafo},
\begin{equation*}
 L+ a^*(v_\Lambda) + a(v_\Lambda)= (1-G_{T_\Lambda})^* K_\Lambda (1-G_{T_\Lambda}) + S_\Lambda - c_0 - \frac{m \Lambda}{\pi^2 (2m+1)},
\end{equation*}
with
\begin{equation*}
 S_\Lambda 
 =- a(v_\Lambda)( K_\Lambda^{-1}-L^{-1} ) a^*(v_\Lambda)=a(v_\Lambda)K_\Lambda^{-1} T_\Lambda L^{-1}  a^*(v_\Lambda).
\end{equation*}
Along the lines of~\cite[Sect.3.4]{LaSch18} one can then obtain a renormalisation procedure as follows.
From our estimates on $T$ and $S$ one deduces that, as $\Lambda \to \infty$, $T_\Lambda$ converges to $T$, and  $S_\Lambda - c \gamma_m \log\Lambda$ converges to $S$ (for some $c\in \R$), strongly as operators from $D(T)$, respectively $D(S)$, to $\hilb$.
Together with similar convergence results for $G_\Lambda$, $G_{T_\Lambda}$, one then obtains strong resolvent convergence of
\begin{equation*}
  L+ a^*(v_\Lambda) + a(v_\Lambda)+\frac{m \Lambda}{\pi^2 (2m+1)} - c \gamma_m \log\Lambda 
\end{equation*}
to the operator $H$.
This shows exactly the divergence in $\Lambda$, with a linear and a logarithmic term, observed numerically in~\cite{grusdt2015} (confirming this observation in view of earlier work~\cite{vlietinck2015}, where convergence without the logarithmic term was claimed). Additionally, our analysis for $M>1$ shows that the logarithmic term is proportional to $M$, even though one might naively expect it to be of order $M^2$, given the form of $S_\Lambda$.
\end{rem}

\appendix
\section{Technical Lemmas}

In this appendix we spell out the details concerning the bounds on $T$ and $S$.
These bounds are obtained using variants of the Schur test, similar to those derived in~\cite{MoSe17}, for sums of integral operators that give control on the growth in $n$ as the number of summands increases. Applying the basic Schur test to every summand would yield a bound that grows like the number of summands. In the following lemma we use the symmetry of the functions in $\hilb\uppar{n}$ to obtain an improvement that is reflected in the order of the sum and the supremum in the constants $\Lambda$, $\Lambda'$ below. In the cases relevant to us, this will lead to bounds that are independent of the number of summands. We also remark that the same lemma holds for antisymmetric wavefunctions, since only the symmetry of $\abs{\psi}^2$ is used.

\begin{lem}\label{lem:Schur}
 Let $\ell < n$ and $d$ be positive integers, $M\in \N$, and
 \begin{equation*}
  \mathscr{J}:=\left\{ J:\{1,\dots,\ell\}\to \{1,\dots,n\}\Big\vert J \text{ one-to-one}\right\}.
 \end{equation*}
For every $J\in \mathscr{J}$ let $\kappa_J\in L^1_\mathrm{loc}(\R^{dM}\times \R^{dn}\times \R^{d\ell})$ be a real, non-negative function and let $F:\R^{d\ell}\times \R^{d\ell}\to \R^{dM}$ be measurable. 

Define the operator $I:\mathscr{D}(\R^{d(M+n)})\to \mathscr{D}'(\R^{d(M+n)})$ by
 \begin{equation*}
  (I\psi)(P,Q)=\sum_{J\in \mathscr{J}} \int_{\R^{d\ell}} \kappa_J(P,Q,R) \psi(P+F(Q_J,R),\hat Q_J, R) \ud R,
 \end{equation*}
where $Q_J=(q_{J(1)},\dots,q_{J(\ell)})=(q_{j_1},\dots, q_{j_\ell})$, and $\hat Q_J$ denotes the vector in $\R^{d(n-\ell)}$ formed by $q_1, \dots, q_n\in \R^d$ without the entries of $Q_J$. Denote by $\kappa_J^t(Q,R)$ the kernel obtained from $\kappa_J(Q,R)$ by exchanging $r_i$ with $q_{j_i}$ for $i=1, \dots,\ell$.

If there exists a positive function $g\in L^\infty_\mathrm{loc}(\R^d)$ for which the quantities
\begin{equation*}
\Lambda:= \sup_{Q\in \R^{d\ell}}\sum_{J\in \mathscr{J}}   
\int_{\R^{dn}} \prod_{i=1}^\ell  \frac{g(q_{j_i})}{g(r_i)}\sup_{P\in \R^{dM}} \kappa_J (P,Q,R)  \ud R
\end{equation*}
and
\begin{equation*}
\Lambda':= \sup_{Q\in \R^{dn}}\sum_{J\in \mathscr{J}} 
\int_{\R^{d\ell}}\prod_{i=1}^\ell  \frac{g(q_{j_i})}{g(r_i)} \sup_{P\in \R^{dM}}\kappa_J^t(P,Q,R)  \ud R
\end{equation*}
are finite, then $I$ extends to a bounded operator from $\hilb\uppar{n}=L^2(\R^{dM})\otimes L^2(\R^{d})^{\otimes_{\mathrm{sym}} n}$ to $L^2(\R^{d(M+n)})$ with norm at most $\sqrt{\Lambda\Lambda'}$.
\end{lem}
\begin{proof}
Since $\kappa_J$ is non-negative, we have for any $\phi,\psi \in \mathscr{D}(\R^{d(M+n)})$ with $\phi,\psi\in \hilb\uppar{n}$ and any $\delta>0$
\begin{align*}
 0\leq \sum_{J\in \mathscr{J}} \int& \left \vert\delta \phi(P,Q)\prod_{i=1}^\ell\sqrt{\frac{g(q_{j_i})}{ g(r_i)}} - \frac1\delta \psi(P+F(Q_J,R),\hat Q_J,R)
\prod_{i=1}^\ell\sqrt{\frac{ g(r_i)}{g(q_{j_i})}}
\right\vert^2  \\
&\times \kappa_J(P,Q,R)\ud P\ud Q \ud R.
\end{align*}
After expanding the square, the quadratic term in $\phi$ can be estimated by
\begin{align*}
   \delta^2\int \bigg(\sum_{J\in \mathscr{J}}\int \kappa_J(P,Q,R)\prod_{i=1}^\ell\frac{g(q_{j_i})}{g(r_i)}\ud R \bigg) \vert \phi(P,Q) \vert^2 \ud P\ud Q \leq \delta^2 \Lambda \Vert \phi \Vert^2_{L^2}.
\end{align*}
For the term with $\abs{\psi}^2$, changing variables to $\hat Q'_J=\hat Q_J$, $Q'_J=R$, $R'=Q_J$, and
using the permutation-symmetry of $\psi$, gives
\begin{align*}
 & \sum_{J\in \mathscr{J}} \int \kappa_J(P,Q,R) \vert\psi(P+F(Q_J,R),\hat Q_J,R) \vert^2 \prod_{i=1}^\ell\frac{g(r_i)}{  g(q_{j_i})}\ud P\ud Q \ud R\\
  &= \sum_{J\in \mathscr{J}} \int \kappa_J^t(P,Q',R') \vert\psi(P+F(R',Q_J'), Q') \vert^2 \prod_{i=1}^\ell \frac{ g(q'_{j_i})}{ g(r_i')}\ud P\ud Q' \ud R'.
 \end{align*}
 Using first the Hölder inequality in $P$ and then changing variables to $P'=P+F(R',Q_J')$, we can bound this by
 \begin{align*}
  &\sum_{J\in \mathscr{J}} \int \left(\sup_{P\in \R^{dM}} \kappa_J^t(P,Q',R')\right) \vert\psi(P', Q') \vert^2 \prod_{i=1}^\ell \frac{ g(q'_{j_i})}{ g(r_i')}\ud P'\ud Q' \ud R'\\
  &\leq \Lambda' \norm{\psi}_{\hilb\uppar{n}}^2.
   \end{align*}

Together, these estimates imply that
\begin{align*}
  &2 \mathrm{Re}(\langle \phi, I \psi \rangle) \\
  &= \sum_{J\in \mathscr{J}} \int \kappa_J(P,Q,R) 2 \mathrm{Re}(\overline{\phi}(P,Q)\psi(P+F(Q_J,R),\hat Q_J,R)) \ud P\ud Q \ud R \\
  &\leq \delta^2 \Lambda \Vert \phi \Vert^2_{L^2} + \frac{\Lambda'}{\delta^2}\Vert \psi \Vert^2_{\hilb\uppar{n}}.
\end{align*}
The same holds true for the negative of the real part of $\langle \phi, I\psi \rangle$, by replacing $\psi$ by $-\psi$, and the imaginary part, replacing $\psi$ by $\ui \psi$.
This yields
\begin{equation*}
\sup_{\substack{\|\phi\|_{L^2}=1 \\ \|\psi\|_{\hilb\uppar n}=1}} |\langle \phi, I \psi\rangle | \leq \frac{1}{ 2} \sqrt{2(\delta^2 \Lambda^2 + \Lambda'^2/\delta^2)} ,
\end{equation*}
so $I$ is bounded from $\hilb\uppar{n}$ to $L^2(\R^{d(M+n)})$. Choosing $\delta=\sqrt{\Lambda'/\Lambda}$ gives the bound on the norm $\norm{I}\leq \sqrt{\Lambda\Lambda'}$.
\end{proof}

Since the operator $T$ of Lemma~\ref{lem:T} is not bounded on $\hilb\uppar{n}$ we need to slightly adapt the technique of Lemma~\ref{lem:Schur} for this case. 

\begin{lem}\label{lem:T_od}
There exists a constant $C$ such that for all $n\in \N$ and $\psi\in \hilb\uppar{n} \cap H^{1}(\R^{3(M+n)})$
 \begin{equation*}
\norm{T_\mathrm{od}\psi}_{\hilb\uppar{n}} \leq C  \norm{\psi}_{H^1(\R^{3(M+n)})}.
 \end{equation*}
\end{lem}
\begin{proof}
 Recall the definition of $T_\mathrm{od}$ from Eq.~\eqref{eq:Tod}
 \begin{align}
  \widehat {T_\mathrm{od} \psi}(P,K)
  =-\frac{1}{(2\pi)^{3}}\bigg( 
&\sum_{\mu,\nu=1}^M \sum_{i=1}^n 
\int \frac{\hat \psi(P-e_\mu \xi + e_\nu k_i,\hat K_i,\xi)}
{L(P-e_\mu \xi, K,\xi)} \ud \xi \label{eq:T_od 1}\\
& + \sum_{\mu\neq \nu=1}^M \int \frac{\hat \psi(P-e_\mu \xi + e_\nu \xi),K)}
{L(P-e_\mu \xi, K,\xi)}\ud \xi\bigg).\label{eq:T_od 2}
 \end{align}
 Since we are not interested in the exact dependence of the norm of $T_\mathrm{od}$ on $M$ and $m$ we will just estimate the operator for fixed indices $\mu, \nu$ and $m=\tfrac12$. Set $\kappa(K,\xi)=\frac{1}{n+1+K^2+\xi^2}$.
 To bound the sum over $i$, we argue as in Lemma~\ref{lem:Schur} and obtain for $\psi\in \hilb\uppar{n} \cap H^{1}(\R^{3(M+n)})$, $\phi \in \hilb\uppar{n}$  and $0<\eps<\frac12$ the inequality
\begin{align*}
 &\sum_{i=1}^n  \int \frac{2\mathrm{Re}\left(\overline \phi(P,K) \hat \psi(P-e_\mu \xi + e_\nu k_i),\hat K_i,\xi)\right)}
{n+1+ (P-e_\mu \xi )^2 +K^2 + \xi^2}\ud P \ud K \ud \xi\\
&\leq  \sum_{i=1}^n \int \kappa(K,\xi)
\begin{aligned}[t]
\bigg(&\kappa(K,\xi)^{-1/2+\eps} |\hat \psi|^2(P-e_\mu \xi + e_\nu k_i),\hat K_i,\xi) 
\frac{(\frac1n+\xi^2)^{1+\eps}}{(1+k_i^2)}\\
&+ \kappa(K,\xi)^{1/2-\eps} | \phi|^2(P,K) \frac{(1+k_i^2)}{{(\frac1n+\xi^2)}^{1+\eps}} \bigg)\ud P \ud K \ud \xi.
\end{aligned}
\end{align*}
The term with $|\phi|^2$ is bounded by
\begin{align*}
 & \norm{\phi}^2_{\hilb\uppar{n}} \sup_{K\in \R^{3n}}   \sum_{i=1}^n (1+k_i^2)\int \kappa(K,\xi)^{3/2-\eps} (1+\xi^2)^{-1-\eps} \ud \xi\\
 &\leq   \norm{\phi}^2_{\hilb\uppar{n}} \sup_{K\in \R^{3n}} \frac{n+K^2}{(n+1 + K^2)} \int \frac{1}{(1+\eta^2)^{3/2-\eps}\eta^{2(1+\eps)}} \ud \eta\\
 &\leq  \Lambda \norm{\phi}^2_{\hilb\uppar{n}} ,
\end{align*}
where $\Lambda$ is clearly independent of $n$.
In the $i$-th term with $|\hat \psi|^2$ we perform the change of variables $(k_i, \xi) \mapsto (\eta, k_i)$ which gives a bound by
\begin{align*}
 & \sum_{i=1}^n \int |\hat\psi |^2(P,K)  \kappa(K,\eta)^{1/2+\eps}\frac{(\frac1n+k_i^2)^{1+\eps}}{1+\eta^2} \ud P \ud K \ud \eta\\
 &\leq  \sum_{i=1}^n \norm{\sqrt{1+K^2}\psi}^2_{\hilb\uppar{n}} \sup_{K\in \R^{3n}}  
 \frac{\sum_{i=1}^n(\frac1n + k_i^2)^{1+\eps}}{(1 + K^2)^{1+\eps}} 
 \int \frac{1}{(1+\eta^2)^{1/2+\eps}\eta^{2}} \ud \eta\\
 &\leq  \Lambda' \norm{\psi}^2_{H^1}. 
\end{align*}
Here $\Lambda'$ is independent of $n$ because the $(1+\eps)$-norm of the vector $(1/n+q_1^2,...,1/n+q_n^2)\in \R^n$ is bounded by its $1$-norm.
By the argument of Lemma~\ref{lem:Schur} this proves that the sum over $i$ in $T_\mathrm{od}$,~\eqref{eq:T_od 1}, defines an operator that is bounded from $H^1$ to $\hilb\uppar{n}$ by $\sqrt{\Lambda \Lambda'}$, which is independent of $n$. 

The remaining operator~\eqref{eq:T_od 2} we have to estimate is
\begin{align*}
 \psi\mapsto \int \frac{\hat \psi(P-e_\mu \xi + e_\nu \xi),K)}
{n+1+(P-e_\mu \xi )^2 +K^2 + \xi^2}\ud \xi
\end{align*}
with $\mu\neq \nu$. Since the number of these terms is independent of $n$ the necessary bound can be obtained by the standard Schur test.
Explicitly, we have, as above,
\begin{align*}
& \int \frac{2\mathrm{Re}\left(\overline{\phi}(P,K) \hat \psi(P-e_\mu \xi + e_\nu \xi),K)\right)}
{n+1+(P-e_\mu \xi )^2 +K^2 + \xi^2}\ud \xi \ud P \ud K\\
&\leq \int \frac{1}{1+(P-e_\mu \xi )^2 + \xi^2}
\bigg(
\begin{aligned}[t]
&|\hat\psi|^2(P-e_\mu\xi + e_\nu \xi, K) \frac{(1+(p_\nu+\xi)^2)^{1+\eps}}{(1+p_\mu^2)^{1/2+\eps}}\\
&+\abs{\phi}^2(P,K) \frac{(1+p_\mu^2)^{1/2+\eps}}{(1+(p_\nu+\xi)^2)^{1+\eps}}\bigg)\ud \xi \ud K \ud P.
 \end{aligned}
\end{align*}
The term with $\phi^2$ is bounded using
\begin{align*}
 &\int \frac{\abs{\phi}^2(P,K)}{1+(P-e_\mu \xi )^2 + \xi^2} \frac{(1+p_\mu^2)^{1/2+\eps}}{(1+(p_\nu+\xi)^2)^{1+\eps}} \ud \xi \ud P \ud K\\
 &\leq \norm{\phi}_{\hilb\uppar{n}}^2 \sup_{P\in \R^{3M}} (1+p_\mu^2)^{1/2+\eps} \int \frac{\ud \xi}{\left(1+\frac12 p_\mu^2 +2 (\xi-\frac12 p_\mu)^2\right)(1+(p_\nu+\xi)^2)^{1+\eps}}  \\
 &\leq \norm{\phi}_{\hilb\uppar{n}}^2 \sup_{P\in \R^{3M}} (1+p_\mu^2)^{1/2+\eps} \int \frac{\ud \xi}{\left(1+\frac12 p_\mu^2 +2 \xi^2\right)(1+\xi^2)^{1+\eps}}\\
&\leq \norm{\phi}_{\hilb\uppar{n}}^2 2\int \frac{\ud \eta}{(1+4\eta^2)\eta^{2+2\eps}},
\end{align*}
where we have used the Hardy-Littlewood inequality. 
After changing variables to $Q=P-e_\mu \xi + e_\nu \xi$ the argument for the term with $|\psi|^2$ is essentially the same, and we conclude as before.
\end{proof}

\begin{lem}\label{lem:S exists}
 Let $\eps>0$, $\psi\in \hilb\uppar{n} \cap H^\eps(\R^{3(M+n)})$ and $R\psi=- L^{-1}T G \psi$. Then 
 \begin{align*}
 S_\mathrm{sing}\psi:= \lim_{r\to 0}\sum_{\mu=1}^M
\frac{1}{4\pi}\int\limits_{S^2}
\left(\sqrt{n+1}\left(R\psi\right)(X,Y,x_\mu+r\omega)+\gamma_m \log(r)\psi(X,Y)\right) \ud \omega
 \end{align*}
 exists in $\hilb\uppar{n}$. 
\end{lem}
\begin{proof}
Let $R_\mathrm{d}=- L^{-1}T_\mathrm{d} G$, $R_\mathrm{od}=- L^{-1}T_\mathrm{od} G$.

In this proof we will focus on the exact form of the logarithmic divergence and discard some regular terms. Precise estimates of these will be given in the proof of Lemma~\ref{lem:S bound}.

When expanded, the expression for $\sqrt{n+1} R_\mathrm{d}=- \sqrt{n+1} L^{-1}T_\mathrm{d} G$ will contain a sum over
 $\nu\in \{1,\dots, M\}$ coming from $T_\ud$, see Eq.~\eqref{eq:Tdiag}, and a second sum over pairs $(\lambda,i)\in \{1,\dots, M\}\times \{1,\dots n+1\}$ coming from the creation operator in $G$. We then want to evaluate this expression on the plane $\{x_\mu=y_{n+1}\}$. For a fixed set of indices, we can write the corresponding term using the Fourier transform, obtaining an expression similar to~\eqref{eq:R_d}. This is, up to a prefactor,
\begin{align}\label{eq:R_d n}
 \int   \frac{\sqrt{n+2+ \frac{1}{2m +1} p_\nu^2 + \frac{1}{2m} \hat P_\nu^2+ K^2}}{L(P,K)^2 }\ue^{\ui X P + \ui Y K} \hat \psi(P+  e_\lambda k_i, \hat K_i) \ud P \ud K.
\end{align}
 To analyse the behaviour as $x_\mu-y_{n+1}\to 0$ it is instructive to set $Q=P+e_\mu k_{n+1}$.
The function $\psi$ then appears as
 \begin{equation*}
 \hat \psi(Q-e_\mu k_{n+1} + e_\lambda k_i, \hat K_{i}).
 \end{equation*}
The operator defined by~\eqref{eq:R_d n} is of a very different nature for $(\mu,n+1)=(\lambda,i)$ and all other cases. In the case of equality, it is essentially a Fourier multiplier by the value of the $k_{n+1}$-integral, since then $\hat\psi$ no longer depends on this variable. This is singular as $x_\mu-y_{n+1}\to 0$, because the integral is not absolutely convergent. One the other hand, if  $(\mu,n+1)\neq(\lambda,i)$ then this is an integral operator that can be bounded on $L^2$ by the Schur test and depends continuously on $x_\mu-y_{n+1}$. We prove uniform bounds in the particle number on this operator in Lemma~\ref{lem:S bound}. 
Consider now the singular terms, with $(\mu,n+1)=(\lambda,i)$. For given $\mu$, there are still $M$ of  these, indexed by $\nu$. 
For $\nu\neq \mu$ the the integral over $k_{n+1} $  can be rewritten as
\begin{align*}
\int \frac{\ue^{\ui \rho(y_{n+1}-x_\mu)}
\sqrt{n+2+\frac{1}{2m+1}\sigma^2 + \frac{2m+1}{2m} \rho^2 + \frac{1}{2m+1}p_\nu^2 + \frac{1}{2m}\hat P_{\mu,\nu}^2   + \hat K_{n+1}^2}}
{(n+1 +\frac{1}{2m+1}\sigma^2 + \frac{2m+1}{2m} \rho^2+ \frac{1}{2m}\hat P_\mu^2 + \hat K^2_{n+1})^2}\ud \rho.
\end{align*}
This can be analysed starting from Eq.~\eqref{eq:asym T^d}, by replacing $\sigma^2 $ with the appropriate expression and the prefactor $\sqrt{(2m+2)/(2m+1)}$ by $\sqrt{(2m+1)/(2m)}$. The term for $\mu=\nu$ has the same prefactor as~\eqref{eq:asym T^d}, so the total prefactor of the divergent term $\log \abs{x_\mu-y_{n+1}}$ is
\begin{equation}\label{eq:alpha M diag}
 -\frac{1}{(2\pi)^3}\left(\frac{2m}{2m+1}\right)^{3}\left(\frac{ 2\sqrt{m(m+1)}}{2m+1} + (M-1) \right).
\end{equation}

For $R_\mathrm{od}$ we have a sum over $(\nu,\lambda,i)\in \{1,\dots, M\}^2\times\{1,\dots,n+2\}$ with $(\nu,n+2)\neq (\lambda,i)$ coming from $T_\mathrm{od}$, Eq.~\eqref{eq:Tod}, and a sum over $(\omega,j)$, with $\omega\in \{1,\dots, M\}$ and $i\neq j\in \{1,\dots, n+2\}$, coming from the creation operator. All the summands can be written in a form similar to~\eqref{eq:R od1}, proportional to
\begin{align}\label{eq:R_od reg}
 \int \frac{\ue^{\ui P X+\ui \hat K_{n+2} \hat Y_{n+2} }\hat\psi( P-e_\nu k_{n+2}+e_\lambda k_i+e_\omega k_j, \hat K_{i,j}) }
{L( P,\hat K_{n+2})L(P-e_\nu k_{n+2},K)L(P-e_\nu k_{n+2}+e_\lambda k_i,\hat K_{i})}
\ud P \ud K
,
\end{align}
where $K=(k_1,\dots, k_{n+2})$ and $L$ denotes the operator on the space with the number of particles corresponding to the dimension of its argument (which is either $n+1$ or $n+2$).
Here, $\hat\psi $ occurs in the form (with $Q=P+e_\mu k_{n+1}$)
 \begin{equation*}
  \psi(Q-e_\mu k_{n+1} - e_\nu k_{n+2} + e_\lambda k_i+ e_\omega k_j, \hat K_{i,j}).
 \end{equation*}
  As in the case $n=0$, the corresponding operator is singular only when the first argument of $\psi$ equals $Q$. Since $(\nu,n+2)\neq (\lambda,i)$ this only happens for $(\nu,{n+2})=(\omega,j)$ and $(\mu,{n+1})=(\lambda,i)$.
 In the singular case, the operator is a Fourier multiplier by a function  of $(Q, \hat K_{n+1,n+2})$ proportional to
  \begin{align*}
 \int \begin{aligned}[t]
         &\frac{1}{n+1+\frac{1}{2m}(Q - e_\mu k_{n+1})^2 + \hat K^2_{n+2}}
         \frac{1}{n+1+\frac{1}{2m}(Q -e_\nu k_{n+2})^2 + \hat K^2_{n+1}}
         \\
         &\times\frac{1}{n+2+\frac{1}{2m}(Q - e_\mu k_{n+1}-e_\nu k_{n+2})^2 + K^2}
         \ud k_{n+1} \ud k_{n+2}.
       \end{aligned}
 \end{align*}
 
 For fixed $\mu$ this gives one term with $\mu=\nu$ that behaves exactly like the expression~\eqref{eq:R od1} for $M=1$, $n=0$. 
 The $M-1$ terms with $\mu\neq \nu$ give a prefactor of $\log (r)$ that exactly cancels the $(M-1)$-term in Eq.~\eqref{eq:alpha M diag}, so $\sqrt{n+1}R\psi$ has the same singularity for $x_\mu-y_{n+1}\to0$ as in the case $M=1$, $n=0$ that was treated in the proof of Proposition~\ref{prop:A_T}. 
 
This proves that the sum of the singular terms and $\gamma_m \log(r) \psi$ converges as $r\to 0$. Convergence of the remaining terms is implied by the bounds of Lemma~\ref{lem:S bound}, as argued in Lemma~\ref{lem:A}.
%
\end{proof}

\begin{lem}\label{lem:S bound}
Let $S_\mathrm{reg}$ be given by~\eqref{eq:S reg} and $S_\mathrm{sing}=S-S_\mathrm{reg}$ be the operator from Lemma~\ref{lem:S exists}.
  There exists a constant $C$ such that for all $\eps>0$, $n\in \N$ and $\psi\in \hilb\uppar{n}\cap H^\eps(\R^{3(M+n)})$  
 \begin{equation}
  \norm{S\psi}_{\hilb\uppar{n}} \leq C\left(\frac1\eps \norm{\psi}_{H^\eps} + (1+\log(n+1)) \norm{\psi}_{\hilb\uppar{n}}\right).
 \end{equation}
\end{lem}
\begin{proof}
For the regular part $S_\mathrm{reg}$ we have
\begin{align*}
\norm{S_\mathrm{reg}\psi}_{\hilb\uppar{n}}
 \leq M &\norm{a(\delta_{x_1})L^{-1}}_{\hilb\uppar{n+1}\to\hilb\uppar{n}}\\
 & \times \norm{\left((T+c_0)L^{-1}(T+c_0) G_T-c_0 G \right)\psi}_{\hilb\uppar{n+1}}.
\end{align*}
By Lemma~\ref{lem:T}, and Proposition~\ref{prop:G_T}, Lemma~\ref{lem:G bound} we have
\begin{equation*}
 \norm{\left((T+c_0)L^{-1}(T+c_0) G_T-c_0 G \right) \psi}_{\hilb\uppar{n+1}} \leq C (n+1)^{-1/4} \norm{\psi}_{\hilb\uppar{n}},
\end{equation*}
and for $0< s\leq 1/2$ we can estimate by the Cauchy-Schwarz inequality
\begin{align*}
 \norm{a(\delta_{x_1})L^{-1}}_{\hilb\uppar{n+1}\to\hilb\uppar{n}}
 &\leq c\sqrt{n+1} \left( \int_{\R^3}\frac1{(n+1+K^2)^{2}}\ud k_{n+1}\right)^{1/2}\\
 &\leq C \frac{(n+1)^{1/4 +s}}{\sqrt s}.
\end{align*}
Choosing $s=1/\log(n+1)^2$ for $n\geq 1$ then gives the bound
\begin{equation*}
\norm{S_\mathrm{reg}\psi}_{\hilb\uppar{n}}  \leq C(1+ \log (n+1))\norm{\psi}_{\hilb\uppar{n}} .
\end{equation*}

For the singular part $S_\mathrm{sing}$ we give quantitative improvements on the proofs of Proposition~\ref{prop:A_T} and Lemma~\ref{lem:S exists}.
Since we are not interested in the exact dependence on $m$ here, we set the mass of the $x$-particles to $m=\tfrac12$ during this proof.

We first estimate the errors made by the simplifications in the calculation of the singularity for $M=1$, $n=0$ in Proposition~\ref{prop:A_T}. These generalise to the corresponding calculations for arbitrary $M, n$ in Lemma~\ref{lem:S exists} in a straightforward way. The replacement made from Eq.~\eqref{eq:R_d} to Eq.~\eqref{eq:asym T^d} produces an error given by
\begin{equation}\label{eq:R_d err}
 c \int \frac{\ue^{\ui \sigma s +\ui \rho r}\left(\sqrt{1+\frac32 \rho^2 + \frac38 \sigma^2 - \frac12 \rho \sigma}-\sqrt{\frac32} \abs{\rho} \right)}{(1+\frac12 \sigma^2 + 2 \rho^2)^2}\hat \psi(\sigma)\ud \sigma \ud \rho.
\end{equation}
The integrand is bounded by
\begin{equation*}
\frac{1+ \tfrac{3}{8}\sigma^2 + \frac12 \sigma \rho}{(1+ \tfrac{3}{8}\sigma^2 + \frac12 \sigma \rho + \tfrac32\rho^2)^{1/2}(1+\frac1{2}\sigma^2 +2\rho^2)^2}
  \leq C \frac{1+|\sigma|^\eps}{(1+\frac1{2}\sigma^2 +2\rho^2)^{3/2+\eps/2}},
\end{equation*}
for $0<\eps<1/2$.
This is integrable in $\rho$, so~\eqref{eq:R_d err} can be evaluated at $r=0$, leading to an estimate by
\begin{equation*}
 \Vert \eqref{eq:R_d err} \vert_{r=0}\Vert_{\hilb\uppar{0}} \leq \frac C \eps \Vert \psi \Vert_{H^\eps}.
\end{equation*}
The simplified singular part of $R_\ud$,~\eqref{eq:asym T^d} contributes a Fourier multiplier by $\log(1+\frac1{\sqrt 2} \abs{\sigma})$ which can be estimated on $H^\eps$ in the same way. This completes the case of $R_\ud$ for $M=1$, $n=0$. The reasoning for the divergent part and arbitrary $n$ is essentially the same, except that there is a term growing like $\log(n+1)$ due to the $n$-dependence of $L$.

In the calculation for $R_\mathrm{od}$ with $M=1$, $n=0$ we made a simplification in replacing $(\sigma-\xi)^2$ by $\sigma^2 + \xi^2$ in the denominator of~\eqref{eq:R od1}, i.e. replacing
\begin{align}\notag
\tau(\sigma, \rho, \xi)=\frac{1}
 {(1+\frac1{2}\sigma^2+2\rho^2)(1 + (\sigma-\xi)^2 +\xi^2)(2+\xi^2 + \frac1{2}(\sigma-\xi)^2+2(\rho+\frac1{2}\xi)^2)}
\end{align}
by
\begin{align}\notag
\tau_0(\sigma, \rho, \xi)=\frac{1}
 {(1+\frac1{2}\sigma^2+2\rho^2)(1 + \sigma^2 +2 \xi^2)(1+\frac32 \xi^2 +\frac12 \sigma^2 +2(\rho+\frac1{2}\xi)^2)}.
\end{align}
We have, with $0<\eps<\frac12$,
\begin{align}\label{eq:tau diff}
 \abs{\tau-\tau_0} \leq C \frac{|\sigma|^\eps}{(1+2\rho^2)^{1/2+\eps/2}}\frac{1}{(1+\xi^2)^{3/2}} \frac{1}{2+\xi^2+\rho^2}.
\end{align}
And this implies that
\begin{equation*}
 \sup_{\sigma \in \R^3} (1+\sigma^2)^{-\eps/2} \int |\tau -\tau_0|(\sigma, \rho, \xi) \ud \rho \ud \xi \leq \frac{C}{\eps},
\end{equation*}
which gives the desired estimate for the error as an operator from $H^\eps$ to $\hilb\uppar{0}$.
The contribution of the simplified $R_\mathrm{od}$, with $\tau_0$ instead of $\tau$, can be bounded by $\log(1+\abs{\sigma})$, as for $R_\ud$.

For general $M$ and $n$ we still need to bound the evaluations of the terms in $\sqrt{n+1}R\psi$ that are regular at $x_\mu=y_{n+1}$. We will prove that the sum of these terms gives rise to a bounded operator on $\hilb\uppar{n}$, whose norm is a bounded function of $n$. 

For $\sqrt{n+1} R_\ud \psi$, the regular terms are the evaluations of~\eqref{eq:R_d n} with $(\mu,n+1)\neq (\lambda,i)$ at $x_\mu=y_{n+1}$.
Denote by $\vartheta_{\mu,\nu,\lambda,i}\psi$ the Fourier transform of this function, that is
\begin{align*}
 \vartheta_{\mu,\nu,\lambda,i}\psi(P,\hat K_{n+1})=
 \frac{1}{2(2\pi)^4 2^{3/2}}\int \kappa_{\mu,\nu}(P,K) \hat \psi(P+e_\lambda k_i-e_\mu k_{n+1},\hat K_i) \ud k_{n+1},
\end{align*}
with
\begin{align*}
 \kappa_{\mu,\nu}(P,K)
 &=\left\lbrace
 \begin{aligned}
 & \frac{\sqrt{n+2+ \frac{1}{2m +1} p_\nu^2 +\frac{1}{2m}  (\hat P_{\nu}-e_\mu k_{n+1})^2+  K^2}}{L(P-e_\mu k_{n+1},K)^2 } \qquad \mu\neq \nu\\
 & \frac{\sqrt{n+2+ \frac{1}{2m +1} (p_\mu-k_{n+1})^2+  \hat P_{\mu}^2+  K^2}}{L(P-e_\mu k_{n+1},K)^2 }
 \qquad \mu=\nu
 \end{aligned}\right.\\
 &\leq \frac{2}{L(P-e_\mu k_{n+1},K)^{3/2}}.
\end{align*}
Applying, for fixed $\mu, \nu, \lambda$, Lemma~\ref{lem:Schur} with kernel $\kappa_{\mu,\nu}(P,\hat K_{n+1}, k_{n+1})$ and weight function $g(k)=1+k^2$ we obtain
\begin{align*}
\norm{\sum_{i=1}^{n} \vartheta_{\mu,\nu,\lambda,i} \psi}_{\hilb\uppar{n}} 
&\leq \frac{\Vert \psi\Vert_{\hilb\uppar{n}}}{(2\pi)^4 2^{3/2}}\sum_{i=1}^{n} \int \frac{(1 +k_i^2)}{(n+1 + K^2 )^{3/2} |k_{n+1}|^2}\ud k_{n+1}\\
&\leq C \Vert \psi \Vert_{\hilb\uppar{n}}.
\end{align*}
The operators $\vartheta_{\mu,\nu,\lambda,i}$ with $i=n+1$ are bounded by the standard Schur test (see also Lemma~\ref{lem:T_od}).
This gives a bound on the evaluation of the regular terms in $\sqrt{n+1}R_\ud$ that is independent of $n$.

For $\sqrt{n+1} R_\mathrm{od}\psi$, the regular terms are given by~\eqref{eq:R_od reg} with indices $(\nu,n+2)\neq (\omega,j)$ or $(\mu,n+1)\neq (\lambda,i)$ and $x_\mu=y_{n+1}$. 
Their Fourier transforms are 
\begin{align*}
 &\Theta_{\mu,\nu,\lambda,\omega,i,j} \psi(P, \hat K_{n+1,n+2})\\
 &=-\frac{1}{(2\pi)^6} \int
 \begin{aligned}[t]
&\frac{1}{L(P-e_\mu k_{n+1},\hat K_{n+1})L(P-e_\nu k_{n+2}-e_\mu k_{n+1}, K)}\\
&\frac{\hat \psi(P-e_\mu k_{n+1} - e_\nu k_{n+2} + e_\lambda k_i + e_\omega k_j)}{L(P-e_\nu k_{n+2}-e_\mu k_{n+1}+ e_\lambda k_i, \hat K_i)}\ud k_{n+1} \ud k_{n+2}.
 \end{aligned}
 \end{align*}
For fixed $\mu,\nu,\lambda,\omega$, there are $n(n+1)$ of these terms. For $j,i<n+1$  we apply Lemma~\ref{lem:Schur} with $\ell=2$. Let $\kappa_{i,j}$ be the kernel of the operator $\Theta_{\mu,\nu,\lambda,\omega,i,j}$, then
\begin{align*}
 \kappa_{i,j}(P,\hat K_{n+1,n+2},k_{n+1},k_{n+2})&\leq\frac{1}{(n+1+\hat K_{n+2}^2)(n+2+K^2)(n+1+\hat K_i^2)}\\
 &\leq \frac{1}{(n+1+\hat K_{n+2}^2)^{3/2}(n+1+\hat K_{i,n+1}^2)^{3/2}}.
\end{align*}
Choosing again the weight function $g(k)=1+k^2$, Lemma~\ref{lem:Schur} gives 
\begin{align*}
 \norm{\sum_{i=1}^{n}\sum_{i\neq j=1}^{n} \Theta_{\mu,\nu,\lambda,\omega,i,j}\psi }_{\hilb\uppar{n}}
 \leq \frac{\norm {\psi}_{\hilb\uppar{n}}}{(2\pi)^6} \sqrt{\Lambda \Lambda'},
\end{align*}
with 
\begin{align*}
 \Lambda &\leq 
 \sup_Q \sum_{i=1}^n \int 
 \frac{(1+q_i^2)(n-1+\hat Q_i^2)}{(1+\xi^2)(n+1+\hat Q^2+\xi^2)^{3/2} (1+\eta^2) (n+1+\hat Q^2_i+\eta^2)^{3/2}}\ud \eta\ud \xi\\
 & \leq \left(\int \frac{1}{\eta^2 (1+\eta^2)^{3/2}} \ud \eta \right)^2,
 \end{align*}
and the same bound for $\Lambda'$.
If $j=n+2$, $i<n+1$, we apply Lemma~\ref{lem:Schur} with $\ell=1$, $\kappa_i=\kappa_{i,n+2}$ (note that $\kappa_i$ and $F=e_\lambda k_i - e_\mu k_{n+1} - e_\nu k_{n+2}$ depend on the additional variable $\xi=k_{n+2}$, but this changes nothing in the proof of Lemma~\ref{lem:Schur}). This gives the bound
\begin{align*}
&\norm{\sum_{i=1}^{n} \Theta_{\mu,\nu,\lambda,\omega,i,n+2}\psi }_{\hilb\uppar{n}}\\
&\leq \frac{\norm {\psi}_{\hilb\uppar{n}}}{(2\pi)^6}
\sup_Q \int \frac{n+Q^2}{(1+\eta^2)(n +Q^2+\eta^2)(n+Q^2+\eta^2+\xi^2)(1+\xi^2)}\ud \eta \ud \xi\\
&\leq \frac{\norm {\psi}_{\hilb\uppar{n}}}{(2\pi)^6}\left(\int \frac{1}{\eta^2 (1+\eta^2)^{3/2}} \ud \eta \right)\left(\int \frac{1}{\xi^2 (1+\xi^2)} \ud \xi \right).
\end{align*}
The estimate for the sum with $i=n+1$, $j<n+1$ is the same.
%
The remaining operators with $(i,j)=(n+1, n+2)$ (but restrictions on $\mu,\nu,\lambda,\omega$) are again bounded by the usual Schur test.
This completes the proof of the lemma.
\end{proof}

\begin{lem}\label{lem:S sym}
 For any $\eps>0$ and $n\in \N$ the operator $S$ is symmetric on the domain $D(S)=\hilb\uppar{n} \cap H^\eps(\R^{3(M+n)})$.
\end{lem}
\begin{proof}
The operator $S_\mathrm{reg}$, defined in~\eqref{eq:S reg}, can be written as
\begin{equation*}
 -\sum_{\mu=1}^M\sum_{\nu=1}^M a(\delta_{x_\mu}) L^{-1} \left( (T+c_0) L^{-1} (T+c_0)  (L+T+c_0)^{-1} - c_0 L^{-1}\right) a^*(\delta_{x_\nu}).
\end{equation*}
The operator $aL^{-1}:\hilb\uppar{n+1}\to \hilb\uppar{n}$ is bounded and $(aL^{-1})^*=L^{-1}a^*$, so the second term above is bounded and symmetric. For the first term, observe additionally that
\begin{equation*}
  L^{-1} (T+c_0) L^{-1} (T+c_0)  (L+T+c_0)^{-1}=(L+T+c_0)^{-1} (T+c_0) L^{-1} (T+c_0)L^{-1},  
\end{equation*}
by the resolvent formula. This implies that  $S_\mathrm{reg}$ is bounded and symmetric.

As shown in Lemma~\ref{lem:S exists}, the divergent terms in $S_\mathrm{sing}$ give rise to real Fourier multipliers. These are bounded from $H^\eps(\R^{3(M+n)})$ to $\hilb\uppar{n}$ by Lemma~\ref{lem:S bound}, and thus symmetric on this domain. The regular terms in $S_\mathrm{sing}$ give rise to a bounded operator by the proof of Lemma~\ref{lem:S bound}. 
Since these terms are regular, the limit in the definition of $S_\mathrm{sing}$ just gives the evaluation at $y_{n+1}=x_\mu$. Denote this evaluation map by $\tau_{x_\mu}(y_{n+1})$. Then the contribution of the regular part of $R_\ud$ to $S_\mathrm{sing}$ is
\begin{equation*}
 \sum_{\mu=1}^M\sum_{(\lambda,i)\neq (\mu,n+1)}\tau_{x_\mu}(y_{n+1}) L^{-1} T_\ud L^{-1} \delta_{x_\lambda}(y_i) \psi(X,\hat Y_i).
\end{equation*}
This defines a symmetric operator because $(\tau_x L^{-1})^* = L^{-1} \delta_x$, $T_\ud$ is symmetric and both $T_\ud$ and $L$ commute with permutations of the $y_i$. Similarly,
%
the regular terms in $R_\mathrm{od}$ are the sum of
\begin{align*}
\tau_{x_\mu}(y_{n+1}) L^{-1} \tau_{x_\nu}(y_{n+2})  L^{-1} \delta_{x_\lambda}(y_i) L^{-1}\delta_{x_\omega}(y_j)\psi\uppar{n}(X,\hat Y_{i,j})
\end{align*}
over all indices with $i\neq j$ and $(\mu,n+1)\neq (\lambda,i)$ or $(\nu,n+2)\neq (\omega,j)$ (see also Eq.~\eqref{eq:T sym}). This operator is symmetric for the same reason as in the case of $R_\ud$.
\end{proof}

\section*{Acknowledgements}

I am grateful to Stefan Keppeler, Julian Schmidt, Stefan Teufel and Roderich Tumulka for many interesting discussions on the subject of interior-boundary conditions.

\newcommand{\etalchar}[1]{$^{#1}$}
\providecommand{\bysame}{\leavevmode\hbox to3em{\hrulefill}\thinspace}
\providecommand{\MR}{\relax\ifhmode\unskip\space\fi MR }
\providecommand{\MRhref}[2]{%
  \href{http://www.ams.org/mathscinet-getitem?mr=#1}{#2}
}
\providecommand{\href}[2]{#2}

\end{document}